\documentclass[a4paper,10pt]{article}
\pdfoutput=1

\usepackage[utf8]{inputenc}
\usepackage[T1]{fontenc}

\usepackage{a4wide}

\usepackage{amsmath,amssymb}
\usepackage{amsthm}
\usepackage{enumerate} 

\usepackage{authblk}

\usepackage[svgnames]{xcolor}
\definecolor{highlightNEW}{named}{black}

\usepackage{hyperref}
\hypersetup{
    colorlinks=true,
    linkcolor=Navy,
    citecolor=Navy,
    urlcolor=Navy
}

\usepackage{graphicx}
\graphicspath{{fig/}}
\usepackage{epstopdf}
\usepackage{subcaption}

\usepackage{booktabs}
\usepackage{multirow}

\usepackage[absolute,overlay,showboxes]{textpos}
\TPMargin{2mm}
\textblockrulecolour{Navy}

\usepackage[round]{natbib}
\bibpunct[, ]{(}{)}{;}{a}{}{,}

\newtheorem{theorem}{Theorem}[section] 
\newtheorem{corollary}[theorem]{Corollary} 
\newtheorem{example}[theorem]{Example} 
\newtheorem{lemma}[theorem]{Lemma} 
\newtheorem{proposition}[theorem]{Proposition} 
\newtheorem{remark}[theorem]{Remark}

\newcommand{\doi}[1]{DOI~\href{\detokenize{http://dx.doi.org/#1}}{\detokenize{#1}}}

\renewcommand{\d}{\,\mathrm{d}}
\newcommand{\e}{\mathrm{e}}
\newcommand{\E}{\mathbb{E}}
\newcommand{\F}{\mathcal{F}}

\newcommand{\R}{\mathbb{R}}

\newcommand{\1}{\mathbf{1}}

\newcommand{\la}{\!\left\langle}
\newcommand{\ra}{\right\rangle}
\newcommand{\p}{\partial}

\newcommand*{\Cdot}[1][1.25]{%
  \mathpalette{\CdotAux{#1}}\cdot%
}
\newdimen\CdotAxis
\newcommand*{\CdotAux}[3]{%
  {%
    \settoheight\CdotAxis{$#2\vcenter{}$}%
    \sbox0{%
      \raisebox\CdotAxis{%
        \scalebox{#1}{%
          \raisebox{-\CdotAxis}{%
            $\mathsurround=0pt #2#3$%
          }%
        }%
      }%
    }%
    \dp0=0pt %
    \sbox2{$#2\bullet$}%
    \ifdim\ht2<\ht0 %
      \ht0=\ht2 %
    \fi
    \sbox2{$\mathsurround=0pt #2#3$}%
    \hbox to \wd2{\hss\usebox{0}\hss}%
  }%
}
\makeatletter
\def\mathcolor#1#{\@mathcolor{#1}}
\def\@mathcolor#1#2#3{%
  \protect\leavevmode
  \begingroup
    \color#1{#2}#3%
  \endgroup
}
\makeatother

\newcommand{\NEW}[1]{\mathcolor{highlightNEW}{#1}}
\let\oldalpha\alpha
\renewcommand{\alpha}{\mathcolor{highlightNEW}{\oldalpha}}
\newcommand{\ccode}[2]{\par
        \vspace*{8pt}
        {{\leftskip18pt\rightskip\leftskip
        \noindent{\it #1}\/: #2\par}}\par}
\newcommand{\keywords}[1]{\ccode{Keywords}{#1}}
\newcommand{\email}[1]{\href{mailto:#1}{#1}}

\def\received#1{Received~#1\par}

\newcommand{\jpTitle}{Decomposition formula for rough Volterra stochastic volatility models}
\newcommand{\jpAuthors}{R. Merino, J. Posp\'{\i}\v{s}il, T. Sobotka, T. Sottinen and J. Vives}
\newcommand{\jpKeywords}{Volterra stochastic volatility; rough volatility; rough Bergomi model; option pricing; decomposition formula}
\newcommand{\jpMSC}{60G22; 91G20; 91G60}
\newcommand{\jpJEL}{G12; C58; C63}
\newcommand{\jpDateReceived}{1 August 2019}
\newcommand{\jpDate}{}

\author[1,4]{Ra\'{u}l Merino} 
\author[2]{Jan Posp\'{\i}\v{s}il\thanks{Corresponding author, \email{honik@kma.zcu.cz}}} 
\author[2]{Tom\'{a}\v{s} Sobotka} 
\author[3]{\authorcr Tommi Sottinen} 
\author[1]{Josep Vives} 
\affil[1]{Facultat de Matem\`{a}tiques i Inform\`{a}tica, Universitat de Barcelona, \authorcr Gran Via 585, 08007 Barcelona, Spain,\vspace*{3pt}}
\affil[2]{NTIS - New Technologies for the Information Society, Faculty of Applied Sciences, \authorcr University of West Bohemia, Univerzitn\'{\i} 2732/8, 301 00 Plze\v{n}, Czech Republic,\vspace*{3pt}}
\affil[3]{Department of Mathematics and Statistics, University of Vaasa, \authorcr P.O. Box 700, FIN-65101 Vaasa, Finland,\vspace*{3pt}}
\affil[4]{VidaCaixa S.A., Investment Risk Management Department, \authorcr C/Juan Gris, 2-8, 08014 Barcelona, Spain.}

\ifpdf
\hypersetup{
  pdftitle={\jpTitle},
  pdfauthor={\jpAuthors},
  pdfkeywords={\jpKeywords},
  pdfinfo={
      MSC={\jpMSC},
      JEL={\jpJEL}
  }
}
\fi

\title{\textcolor{Navy}{\textsc{\jpTitle}}}
\date{\jpDate}

\begin{document}

\maketitle

\begin{center}
\received{\jpDateReceived}
\end{center}

\begin{abstract}
The research presented in this article provides an alternative option pricing approach for a class of rough fractional stochastic volatility models. These models are increasingly popular between academics and practitioners due to their surprising consistency with financial markets. However, they bring several challenges alongside. Most noticeably, even simple non-linear financial derivatives as vanilla European options are typically priced by means of Monte-Carlo (MC) simulations which are more computationally demanding than similar MC schemes for standard stochastic volatility models.

In this paper, we provide a proof of the prediction law for general Gaussian Volterra processes. The prediction law is then utilized to obtain an adapted projection of the future squared volatility -- a cornerstone of the proposed pricing approximation. Firstly, a decomposition formula for European option prices under general Volterra volatility models is introduced. Then we focus on particular models with rough fractional volatility and we  derive an explicit semi-closed approximation formula. Numerical properties of the approximation for a popular model -- the rBergomi model -- are studied and we propose a hybrid calibration scheme which combines the approximation formula alongside MC simulations. This scheme can significantly speed up the calibration to financial markets as illustrated on a set of AAPL options.

\end{abstract}

\keywords{\jpKeywords}
\ccode{MSC classification}{\jpMSC}
\ccode{JEL classification}{\jpJEL}

\clearpage

\section{Introduction}\label{sec:introduction}

It is well known that the main issue of the Black-Scholes model lies in its assumptions about volatility of the modelled asset. Opposed to the model assumptions, the realized volatility time series tend to cluster, depend on the spot asset level and certainly they do not take a constant value within a reasonable time-frame (see e.g. \cite{Cont01}).

To deal with the aforementioned inconsistencies, stochastic volatility (SV) models have been proposed originally by \cite{HullWhite87} and later e.g. by \cite{Heston93}. These models do not only assume that the asset price follows a specific stochastic process, but also that the instantaneous  volatility of asset returns is of random nature as well. Especially, the latter approach by Heston became popular in the eyes of both practitioners and academics. Several modifications of this model have been proposed over the last 20 years: models with jump-diffusion dynamics \citep{Bates96,Duffie00}, with time-dependent parameters \citep{MikhailovNogel03,Elices08,Benhamou10}, with fractionally scaled volatility \citep{ComteRenault98, Alos07, ElEuchRosenbaum19} and models with several aspects combined \citep{PospisilSobotka16amf,BaustianMrazekPospisilSobotka17asmb}.

The original pricing approach of \cite{Heston93} was several times revisited, e.g by \cite{Lewis00}, \cite{Attari04} and \cite{Albrecher07} with focus on semi-closed form Fourier transform solutions by \cite{KahlJackel06}, \cite{Alfonsi10} with respect to Monte-Carlo simulation techniques and last but not least by \cite{Alos12} who introduced an analytical approach to option pricing approximation. This approach improves on the techniques introduced by \cite{HullWhite87} and  shows how an adaptive projection of future volatility can be used to price European options under \cite{Heston93} model. Many other papers generalized this idea, see e.g. \cite{Alos15, MerinoVives15}, or recently also a paper by \cite{MerinoPospisilSobotkaVives18ijtaf}. In this article, we revisit the latter approach and we come up with the approximation technique for SV models with volatility process driven by the fractional Brownian motion. This includes exponential rough fractional volatility models introduced by \cite{Gatheral14,Gatheral18}.
 
Although many SV models have been proposed since the original \cite{HullWhite87} model, it seems that none of them can be considered as the universal best market practice approach. Several models might perform well for calibration to complex volatility surfaces, but can suffer from over-fitting or they might not be robust in the sense described by \cite{PospisilSobotkaZiegler18ee}. Also a model with a good fit to implied volatility surface might not be in-line with the observed time-series properties.

Pioneers of the fractional SV models -- \cite{ComteRenault98}, see also \cite{Comte12} -- assumed the so-called Hurst parameter \footnote{Named after a hydrologist Harold Edwin Hurst, for more information on fractional Brownian motion, see Section \ref{ssec:fBm} and e.g. the article by \cite{Mandelbrot68}} ranged within $H \in (1/2,1)$ which implies that the spot variance evolution is represented by a persistent process, i.e. it would have a long-memory property. In \cite{Alos07}, a mean reverting fractional stochastic volatility model with $H\in (0,1)$ was presented. \cite{Gatheral18} and \cite{Bayer16} came up with a more detailed analysis of rough fractional volatility models which should be consistent with market option prices \citep{Bayer16}, with realized volatility time series and also they could provide superior volatility prediction results to several other models \citep{Bennedsen17}. An approach considering a two factor fractional volatility model, combining a rough term ($H<\frac{1}{2}$) and a persistent one ($H>\frac{1}{2}$), was presented in \cite{Funahashi17}. Recently, also an approximation for target-volatility options under log-normal fractional SABR model was studied by \cite{Alos19} who use the Malliavin caluculus techniques to derive the decomposition formula. In parallel, short-term at-the-money asymptotics for a class of stochastic volatility models were studied by \cite{ElEuchEtAl19} who use the Edgeworth expansion of the density of an asset price.

A typical problem of rough fractional models lies in their tractability -- only cumbersome simulation techniques seem to be available for vanilla European options at the time of writing this article. This serves as a motivation to develop a pricing approximation based on the works of \cite{Alos12}, whose approach was further generalised by \cite{MerinoVives15} and \cite{MerinoPospisilSobotkaVives18ijtaf}. Without using the Malliavin calculus we derive a decomposition and approximation formula for European option prices under general Volterra volatility models and in particular under a class of models with rough fractional volatility. Obtained results hence give a better understanding of the whole surface of prices than results focusing only on at-the-money volatility skew.

The structure of the paper is the following. Section~\ref{sec:preliminaries} is devoted to preliminaries. In Section~\ref{sec:generic_deco} we present a generic decomposition formula of the vanilla European option fair value. In Setion~\ref{sec:model} we present Volterra volatility models. In particular in Section~\ref{ssec:genVolterra} we prove the prediction law for general Gaussian Volterra processes, in Section~\ref{ssec:expVolterra} we obtain the decomposition formula for the exponential Volterra volatility model including the error bound for the approximation formula and in Section~\ref{ssec:fBm} we provide particular results for exponential fracional volatility models (including the standard Brownian motion case). In Section~\ref{sec:numerics}, we provide a numerical comparison of approximated prices against Monte Carlo simulations. All obtained results are concluded in Section~\ref{sec:conclusion}.

\section{Preliminaries and notation}\label{sec:preliminaries}

Let $S=(S_{t}, t\in[0,T])$ be a strictly positive asset price process under a market chosen risk neutral probability measure $\mathcal{P} $ that follows the model:  

\begin{eqnarray}\label{model}
\d S_{t}= r S_{t} \d t + \sigma_{t}S_{t}  \left(\rho \d W_{t}  + \sqrt{1-\rho^{2}}\d \tilde{W}_{t} \right),
\end{eqnarray}
where $S_{0}$ is the current price, $r\geq 0$ is the interest rate, $W_{t}$ and $\tilde{W}_{t}$ are independent standard Wiener processes defined on a probability space $(\Omega, \mathcal{F}, \mathcal{P})$ and $\rho \in \left(-1,1\right)$. In the following, we will denote by $\mathcal{F}^{W}$ and $\mathcal{F}^{\tilde{W}}$ the filtrations generated by $W$ and $\tilde{W}$ respectively. Moreover, we define $\mathcal{F}:=\mathcal{F}^{W} \vee \mathcal{F}^{\tilde{W}}$. The volatility process $\sigma_{t}$ is a square-integrable process assumed to be adapted to the filtration generated by $W$ and its trajectories are assumed to be a.s. c\`adl\`ag and stricly positive a.e.. Note that $\rho$ is the correlation between the price and the volatility processes.

Without any loss of generality, it will be convenient in the following sections to make the change of variable $X_{t} = \log S_{t}, t\in[0,T]$, and write

\begin{eqnarray}\label{log-model}
\d X_{t}= \left(r -\frac{1}{2}\sigma^{2}_{t}\right)\d t + \sigma_{t} \left(\rho \d W_{t}  + \sqrt{1-\rho^{2}} \d \tilde{W}_{t} \right).
\end{eqnarray}

Recall that $Z:= \rho  W  + \sqrt{1-\rho^{2}} \tilde{W}$ is a standard Wiener process.

The following notation will be used throughout the paper:

\begin{itemize}

\item We will denote by $BS(t,x,y)$ the price of a plain vanilla European call option under the classical Black-Scholes model with constant volatility $y$, current log stock price $x$, time to maturity $\tau=T-t$, strike price $K$ and interest rate $r$. In this case, 
\begin{eqnarray}
\nonumber BS\left(t,x, y\right)= e^{x} \Phi(d_{+}) - K e^{-r\tau} \Phi(d_{-}),
\end{eqnarray}
where $\Phi(\cdot)$ denotes the cumulative distribution function of the standard normal law and 
\begin{eqnarray}
\nonumber d_{\pm}(y) = \frac{x-\ln K + (r \pm \frac{y^{2}}{2})\tau}{y\sqrt{\tau}}.
\end{eqnarray}

\item In our setting, the call option price is given by 

$$V_{t}=e^{-r\tau}{\mathbb E}_t [(e^{X_{T}}-K)^+]$$

where ${\mathbb E}_t$ is the conditional expectation respect to the $\sigma-$algebra $\mathcal{F}_t$.

\item  Recall that from the Feynman-Kac formula for the model \eqref{log-model}, the operator 
\begin{eqnarray}{\label{FK}}
\mathcal{L}_{y}:= {\partial}_t + \frac{1}{2}y^{2} {\partial^{2}_{x}} 
					+ \left(r -\frac{1}{2}y^2 \right) {\partial}_{x}-r
\end{eqnarray}
satisfies ${\mathcal L}_{y}BS(t,x,y)=0$.

\item We define the operators $\Lambda:=\partial_x$, $\Lambda^{n}:=\partial^{n}_x$, $\Gamma:=\left(\partial^{2}_{x}-\partial_{x}\right)$ and $\Gamma^{2}=\Gamma\circ \Gamma.$ In particular, for the Black-Scholes formula, using straightforward calculations, we get:
\begin{eqnarray*}
\Gamma BS(t,x,y)&=& \frac{e^{x}}{y\sqrt{2\pi \tau}}\exp\left(-\frac{d_+^2(y)}{2}\right),\\
\Lambda \Gamma BS(t,x,y)&=& \frac{e^{x}}{y \sqrt{2\pi \tau}}\exp\left(-\frac{d_+^2(y)}{2}\right)\left(1-\frac{d_+(y)}{y \sqrt{\tau}}\right),\\
\Gamma^2 BS(t,x,y)&=& \frac{e^{x}}{y \sqrt{2\pi \tau}}\exp \left(-\frac{d_+^2(y)}{2}\right)\frac{d_+^2(y)-y d_+(y) \sqrt{\tau}-1}{y^2\tau}.
\end{eqnarray*}

\item We define
$$R_{t}:=\frac{1}{8}\E_t\left[\int^{T}_{t}\d\la M,M\ra_{u}\right]$$
and 

$$U_{t}:=\frac{\rho}{2}\E_t\left[\int^{T}_{t}\sigma_{u}\d\la M,W\ra_{u}\right],$$
where $\la\cdot, \cdot\ra$ denotes the quadratic covariation process and $M$ is the $\mathcal{F}-$martingale defined by
\begin{equation}\label{definition_M}
M_{t}:=\int^{T}_{0}\E_t\left[\sigma^{2}_{s}\right]\d s.
\end{equation}
\end{itemize}

\section{A generic decomposition formula}\label{sec:generic_deco}

In this section, we provide an insight on a generic decomposition formula based on the work of \cite{Alos12}, \cite{MerinoVives15} and \cite{MerinoPospisilSobotkaVives18ijtaf}. In particular, we recover the results for a generic stochastic volatility model presented in \cite{MerinoPospisilSobotkaVives18ijtaf}.

It is well known that if the stochastic volatility process is independent of the price process, the pricing formula of a vanilla European call option is given by 
$$V_{t}=\E_t [BS(t,X_{t},{\bar \sigma}_{t})]$$  
where ${\bar \sigma}^2_{t}$ is the so called \emph{average future variance} that is defined by 
\begin{eqnarray}
\nonumber {\bar\sigma}^2_{t}:=\frac{1}{T-t} \int^{T}_{t}\sigma^{2}_{u}\d u.
\end{eqnarray}
Naturally, ${\bar\sigma}_{t}$ is called the \emph{average future volatility}.   

We consider the adapted projection of the future variance
\begin{equation}\label{e:a}
a_t^2 := \int^{T}_{t}\E_t[\sigma^{2}_{u}]\d u
\end{equation}
and the average future variance as
\begin{equation*}
v^2_{t}:=\E_t({\bar\sigma}^2_{t}) = \frac{a_t^2}{T-t} 
\end{equation*}
to obtain a decomposition of $V_{t}$ in terms of $v_{t}.$ This idea switches an anticipative problem related to the anticipative process ${\bar\sigma}_{t}$ into a 
non-anticipative one related to the adapted process $v_{t}$. 

Taking into account $M$ defined in \eqref{definition_M}, we can write
$$\d v^{2}_{t}= \frac{1}{T-t}\left[\d M_{t}+\left(v^{2}_{t}-\sigma^{2}_{t}\right)\d t\right].$$

In this paper, we will utilize the following lemma which is proved in \cite{Alos12}, p.~406.
\begin{lemma}\label{lemaclau} 
Let $0\leq t\leq u\leq T$ and $\mathcal{G}_{t}:=\mathcal{F}_{t}\vee
\mathcal{F}_{T}^{W}.$ For every $n\geq 0,$ there exists $C=C(n) $ such that
\begin{equation*}
\left\vert \mathbb{E} \left( \left. {\Lambda^{n}\Gamma BS}\left( u,X_{u},v_{u}\right)
\right\vert \mathcal{G}_{t}\right) \right\vert 
\leq C (a_u^2)^{-\frac12(n+1)},
\end{equation*}
where $a_t^2$ is defined by \eqref{e:a}.
\end{lemma}

\begin{remark}
It is easy to see that the previous Lemma holds for put options and for several other non-path dependent options as well (e.g. Gap options).
\end{remark}

Now we use a generic decomposition formula proved in \cite{MerinoPospisilSobotkaVives18ijtaf}.\par

\allowdisplaybreaks

\begin{theorem}[Generic decomposition formula]\label{Generic Deco}

Let $B_{t}$ be a continuous semimartingale with respect to the filtration $\mathcal{F}^W$, let $A(t,x,y)$ be a $C^{1,2,2} ([0,T]\times [0,\infty)\times [0,\infty))$ function and let $v^2_t$ and $M_t$ be defined as above. Then we are able to formulate the expectation of $e^{-rT}A(T,X_{T},v^{2}_{T})B_{T}$
in the following way:
\begin{eqnarray*}
&&{\mathbb{E}_{t}}\left[e^{-r(T-t)}A(T,X_{T},v^{2}_{T})B_{T}\right]=A(t,X_{t},v^{2}_{t})B_{t}\\
&+&\mathbb{E}_{t}\left[\int^{T}_{t}e^{-r(u-t)}\partial_{y}A(u,X_{u},v^{2}_{u})B_{u}\frac{1}{T-u}\left(v^{2}_{u}-\sigma^{2}_{u}\right)\d u\right]\\
&+&\mathbb{E}_{t}\left[\int^{T}_{t}e^{-r(u-t)}A(u,X_{u},v^{2}_{u})\d B_{u}\right]\\ 
&+&\frac{1}{2}\mathbb{E}_{t}\left[\int^{T}_{t}e^{-r(u-t)}\left(\partial^{2}_{x}-\partial_{x}\right)A(u,X_{u},v^{2}_{u})B_{u}\left(\sigma^{2}_{u}- v^{2}_{u}\right)\d u\right]\\ 
&+&\frac{1}{2}\mathbb{E}_{t}\left[\int^{T}_{t}e^{-r(u-t)}\partial^{2}_{y}A(u,X_{u},v^{2}_{u})B_{u}\frac{1}{(T-u)^{2}} \d\la M_{\Cdot},M_{\Cdot}\ra_{u}\right]\\ 
&+&\rho \mathbb{E}_{t}\left[\int^{T}_{t}e^{-r(u-t)}\partial^{2}_{x,y}A(u,X_{u},v^{2}_{u})B_{u}\frac{\sigma_{u}}{T-u}\d\la W_{\Cdot},M_{\Cdot}\ra_{u}\right]\\ 
&+&\rho \mathbb{E}_{t}\left[\int^{T}_{t}e^{-r(u-t)}\partial_{x}A(u,X_{u},v^{2}_{u})\sigma_{u} \d\la W_{\Cdot},B_{\Cdot}\ra_{u}\right]\\ 
&+&\mathbb{E}_{t}\left[\int^{T}_{t}e^{-r(u-t)}\partial_{y}A(u,X_{u},v^{2}_{u})\frac{1}{T-u} \d\la M_{\Cdot},B_{\Cdot}\ra_{u}\right].
\end{eqnarray*}
\end{theorem}
\begin{proof}
See the proof of Theorem 3.1 in \cite{MerinoPospisilSobotkaVives18ijtaf}.
\end{proof}

Using the previous decomposition formula, we find that 
\begin{corollary}[BS decompostion formula]\label{BS Deco}
Under assumptions of the Theorem \ref{Generic Deco}, we can obtain a decomposition of European option price $V_t$ as:
\begin{align*}
V_{t} &= BS(t,X_{t},v_{t})\\
&\quad+ \frac{\rho}{2}\E_{t}\left[\int^{T}_{t}e^{-r(u-t)}\Lambda\Gamma BS(u,X_{u},v_{u})\sigma_{u} \d\la W_{\Cdot},M_{\Cdot}\ra_{u}\right]\\
&\quad+ \frac{1}{8}\E_{t}\left[\int^{T}_{t}e^{-r(u-t)}\Gamma^{2} BS(u,X_{u},v_{u}) \d\la M_{\Cdot},M_{\Cdot}\ra_{u}\right] \\
&= BS(t,X_{t},v_{t}) +(I)+(II)
\end{align*}
\end{corollary}

\begin{proof}
Using Theorem \ref{Generic Deco} with $A(t,X_{t},v^{2}_{t})=BS(t,X_{t},v_{t})$ and $B\equiv 1$,  the proof follows in a straightforward way.
\end{proof}

The terms $(I)$ and $(II)$ are not easy to evaluate. Therefore, it becomes important to find simpler approximations to $(I)$ and $(II)$ and estimate the error terms.  In order to find these approximations, we are going to apply Theorem \ref{Generic Deco} to find a decomposition formula for the terms $(I)$ and $(II)$. Using 
$$A(t,X_{t}, v^{2}_{t})=\Lambda\Gamma BS(t,X_{t},v_{t})$$ and $$B_{t}=U_{t}=\frac{\rho}{2}\E_{t}\left[\int^{T}_{t}\sigma_{u} \d\la W_{\Cdot},M_{\Cdot}\ra_{u}\right]$$
a decomposition of the term $(I)$ can be found, and using
$$A(t,X_{t}, v^{2}_{t})=\Gamma^{2} BS(t,X_{t},v_{t})$$ and $$B_{t}=R_{t}=\frac{1}{8}\E_{t}\left[\int^{T}_{t} \d\la M_{\Cdot},M_{\Cdot}\ra_{u}\right]$$
a decomposition of the term $(II)$ is obtained.

After that process, we can approximate the price of a call option by 

\begin{align*}
V_{t} &= BS(t,X_{t},v_{t})\\
&\quad +\Lambda\Gamma BS(t,X_{t},v_{t})U_{t}\\
&\quad +\Gamma^{2} BS(t,X_{t},v_{t})R_{t} \\
&\quad + \epsilon_{t}.
\end{align*}
where $\epsilon_{t}$ denotes error terms. Terms of $\epsilon_{t}$ under a general setting for $\sigma_t$ are provided in Appendix \ref{sec:error_terms}. We note that the error term will depend on the assumed volatility dynamics.

\section{Volterra volatility models}\label{sec:model}

\subsection{General Volterra volatility model}\label{ssec:genVolterra}
In this section, we apply the generic decomposition formula to model \eqref{log-model} with \emph{general Volterra volatility process} defined as
\begin{equation} \label{e:volprocess}
\sigma_t := g(t, Y_t), \quad t\geq 0,
\end{equation}
where $g: [0,+\infty) \times \mathbb{R} \mapsto [0,+\infty)$ is a deterministic function such that $\sigma_t$ belongs to $L^{1}(\Omega \times [0,+\infty))$ and $Y = (Y_t, t\geq 0)$ is the Gaussian Volterra process
\begin{equation} \label{e:Volterra}
Y_t = \int_0^t K(t,s)\,\d W_s,
\end{equation}
where $K(t,s)$ is a kernel such that for all $t>0$
\begin{equation}
\int\limits_0^t K^2(t,s) \d s < \infty, \tag{A1}\label{A1}
\end{equation}
and
\begin{equation}
\F^Y_t = \F^W_t. \tag{A2}\label{A2}
\end{equation}
Let 
\begin{equation}\label{e:r_ts}
r(t,s) := \E[Y_t Y_s], \quad t,s\geq 0, 
\end{equation}
denote the autocovariance function of process $Y_t$ and 
\begin{equation}\label{e:r}
r(t) := r(t,t) = \E[Y_t^2], \quad t\geq 0, 
\end{equation}
be the variance function (i.e. the second moment).

Extending the Theorem 3.1 in \cite{SottinenViitasaari2017} enables us to rephrase the adapted projection of the future squared volatility.

\begin{theorem}[Prediction law for Gaussian Volterra processes]\label{t:prediction_law}
Let $(Y_t, t\geq 0)$ be the Gaussian Volterra process \eqref{e:Volterra} satisfying assumptions \eqref{A1} and \eqref{A2}. 
Then, the conditional process $(Y_u|\F_t, 0\leq t\leq u)$ is Gaussian with $\F_{u}$-measurable mean function
\begin{align}
\hat{m}_{t}(u) &:= \E_{t}[Y_u] = \int_0^t K(u,s)\,\d W_s, \\
\intertext{and deterministic covariance function }
\hat{r}(u_1,u_2|t) &:= \E_{t}\left[ \left(Y_{u_1} - \hat{m}_t(u_1)\right) \left(Y_{u_2} - \hat{m}_t(u_2)\right)\right] \notag \\
&= r(u_1,u_2) - \int_0^t K(u_1,v) K(u_2,v)\,\d v 
\end{align}
for $u_1,u_2\geq t$. 
\end{theorem}
\begin{proof}
Let $0\leq t\leq u$. Then
\begin{align*}
\hat{m}_{t}(u) &= \E_{t}[Y_u] 
= \E\Biggl[ \int_0^u K(u,s)\,\d W_s \Bigg| \F_t^W \Biggr] 
= \int_0^t K(u,s)\,\d W_s
\intertext{and}
\hat{r}(u_1,u_2|t) 
&= \E\left[ \left(Y_{u_1}  - \hat{m}_t(u_1)\right) \left(Y_{u_2} - \hat{m}_t(u_2)\right)\right | \F_t^W ] \\
&= \E\Biggl[ \left( \int_0^{u_1} K(u_1,v_1) \d W_{v_1} - \int_0^t K(u_1,v_1)\,\d W_{v_1}\right) \\
&\quad\quad \cdot\left( \int_0^{u_2} K(u_2,v_2) \d W_{v_2} - \int_0^t K(u_2,v_2)\,\d W_{v_2}\right) \Bigg| \F_t^W \Biggr] \\
&= \E\Biggl[ \int_t^{u_1} K(u_1,v_1) \d W_{v_1} \int_t^{u_2} K(u_2,v_2) \d W_{v_2} \Bigg| \F_t^W \Biggr] \\
&= \int_t^{u_1\wedge u_2} K(u_1,v) K(u_2,v)\,\d v \\
&= r(u_1,u_2) - \int_0^t K(u_1,v) K(u_2,v)\,\d v.
\end{align*}
\end{proof}

In the upcoming sections, we will denote $\hat{r}(u|t):=\hat{r}(u,u|t)$.\\
Under the general volatility process \eqref{e:volprocess}, we have
\begin{eqnarray*}
v^{2}_{t}=\frac{1}{T-t}\int^{T}_{t} \E_{t}\left[g^{2}(u,Y_u)\right]\,\d u
\end{eqnarray*}
and the martingale
\begin{eqnarray*}
M_t=\int_0^T \E_t\left[g^{2}(u,Y_u)\right]\,\d u.
\end{eqnarray*}

In the upcoming lemma, we express the conditional expectation of the future squared volatility in terms of the mean function $\hat{m}_t(u)$.

Let us denote: 
\begin{equation*}
F(t,\hat{m}_{t}(u)):=\E_{t}\left[g^{2}(u,Y_u)\right],
\end{equation*}

\begin{lemma}[Auxiliary terms in the decomposition formula for the general volatility model]
Let $0\leq t\leq u$ and $F(t,\hat{m}_{t}(u))=\E_{t}\left[g^{2}(u,Y_u)\right]$, then 
\begin{align}
\d F(t,\hat{m}_{t}(u)) &= \left(\partial_1 F(t,\hat{m}_{t}(u)) + \frac12 \partial_{22} F(t,\hat{m}_{t}(u))\,K^{2}(u,t)\right)\d t \notag \\
&\quad+\partial_2 F(t,\hat{m}_{t}(u)) \d\hat{m}_{t}(u), \label{e:dF} \\
\d\la M_{\Cdot},W_{\Cdot}\ra_{t} &= \int^{T}_{0} \partial_{2} F(t,\hat{m}_{t}(u)) K(u,t) \d u \d t,\\
\d\la M_{\Cdot},M_{\Cdot}\ra_{t} &= \int^{T}_{0} \int^{T}_{0} \partial_2 F(t,\hat{m}_{t}(u_1)) \partial_2 F(t,\hat{m}_{t}(u_2)) \cdot \notag \\
&\qquad\qquad \cdot K(u_1,t)  K(u_2,t) \d u_1 \d u_2 \d t.
\end{align}
\end{lemma}
\begin{proof}
Let $0\leq t\leq u$ and
\begin{equation}\label{eq:X1}
X_t(u) = \E_t\left[ g^{2}(u,Y_u) \right].
\end{equation}
Theorem \ref{t:prediction_law} implies that
\[
X_t(u) = \int_{\R} g^{2}(u,z) \hat{\varphi}_t(u,z)\,\d z,
\]
where 
\begin{equation}\label{eq:varphi}
\hat\varphi_t(u,z) = \frac{1}{\sqrt{2\pi\hat r(u|t)}} \exp\left\{-\frac12\frac{(z-\hat m_t(u))^2}{\hat r(u|t)}\right\}
\end{equation}
is a Gaussian density function with stochastic mean $\hat{m}_{t}(u)$ and deterministic variance $\hat{r}(u|t)$. 
To calculate the quadratic variation, we note that
\begin{equation}\label{eq:varphi-f}
\hat\varphi_t(u,z) = f(t,\hat m_t(u)),
\end{equation}
where
\begin{equation}\label{eq:f}
f(t,m) = \frac{1}{\sqrt{2\pi\hat r(u|t)}}
\exp\left\{-\frac12\frac{(z-m)^2}{\hat r(u|t)}\right\}.
\end{equation}
Since
\begin{equation*}
\d\la\hat m_{\cdot}(u)\ra_t = K^2(u,t)\, \d t,
\end{equation*}
we retrieve the following expression by It\^o's formula,
\begin{align*}
\d f(t,\hat m_{t}(u)) =
&\left(\p_1 f(t,\hat m_t(u))+\frac12 \p_{22} f(t,\hat m_t(u))K^2(u,t)\right)\d t \\
&+ \p_2 f(t,\hat m_t(u))\d\hat m_t(u),
\end{align*}
and consequently,
\begin{equation*}
\d\la f(\cdot,\hat m_{\cdot}(u))\ra_t 
=
\left(\p_{2}f(t,\hat m_t(u)) K(u,t)\right)^2 \d t.
\end{equation*}
Due to:
\begin{equation*}
\p_2 f(t,\hat m_t(u)) = \frac{z - \hat m_t(u)}{\hat r(u|t)}\hat\varphi_t(u,z), 
\end{equation*}
we obtain
\begin{equation*}
\d\la \hat \varphi_{\Cdot}(u,z)\ra_t = 
\left(\frac{z - \hat m_t(u)}{\hat r(u|t)}\hat\varphi_t(u,z)K(u,t)\right)^2 \d t.
\end{equation*}
More generally, we have
\begin{equation*}
\d\la \hat \varphi_{\Cdot}(u,z_1), \varphi_{\Cdot}(u,z_2) \ra_t = 
\frac{(z_1 - \hat m_t(u))(z_2 - \hat m_t(u))}{\hat r^2(u|t)}\hat\varphi_t(u,z_1)
\hat\varphi_t(u,z_2)K^2(u,t) \d t,
\end{equation*}
and consequently,
\begin{eqnarray*}
\d\la X_{\Cdot}(u)\ra_t 
&=&
\d\la \int_\R g^{2}(u,z_1)\hat\varphi(u,z_1)\d z_1,
\int_\R g^{2}(u,z_2)\hat\varphi(u,z_2)\d z_2 \ra_t \\
&=&
\iint_{\R^2} g^{2}(u,z_1)g^{2}(u,z_2) \d\la \hat \varphi(u,z_1), \hat \varphi(u,z_2)\ra_t  
\d z_1 \d z_2 \\
&=&
\iint_{\R^2} g^{2}(u,z_1)g^{2}(u,z_2) 
(\hat m_t(u)-z_1)(\hat m_t(u)-z_2) \cdot \\
&& \quad\cdot \hat\varphi_t(u,z_1) \hat\varphi_t(u,z_2)
\left(\frac{K(u,t)}{\hat r(u|t)}\right)^2 \d z_1 \d z_2\, \d t. 
\end{eqnarray*}

Let $F=F(t,m)$ be a $\mathcal{C}^{1,2}$-function of time $t$ and ``spot'' $m=\hat m_t(u)$ of the prediction martingale. Because of the filtrations are the same, we have, in general,
\begin{equation*}
\d\la F(\cdot, \hat m_{\Cdot}(u)), W_{\Cdot}\ra_t
=
\p_2 F(t, \hat m_t(u)) K(u,t)\, \d t.
\end{equation*}

Now we set
\[
F(t,\hat{m}_{t}(u)) = X_{t}(u)=\int_{\mathbb{R}} g^{2}(u,z) \hat{\varphi}_t(u,z)\,\d z,
\]
and applying the Itô formula, we obtain \eqref{e:dF}. Moreover, 
\begin{eqnarray*}
\d\la M_{\Cdot}, W_{\Cdot}\ra_t
&=&
\d\la \int_0^T F(\cdot,\hat m_{\Cdot}(u))\d u, W_{\Cdot}\ra_t  \\
&=&
\int_0^T \d\la F(\cdot,\hat m_{\Cdot}(u)), W_{\Cdot}\ra_t \d u \\
&=&
\int_0^T \p_2 F(t, \hat m_t(u)) K(u,t)\d u \,\, \d t
\end{eqnarray*}
and
\begin{eqnarray*}
\d\la M_{\Cdot}, M_{\Cdot}\ra_t
&=&
\d\la \int_0^T F(\cdot,\hat m_{\Cdot}(u))\d u, \int_0^T F(\cdot,\hat m_{\Cdot}(u))\d u \ra_t  \\
&=&
\int_0^T\!\!\!\int_0^T \d\la F(\cdot,\hat m_{\Cdot}(u_{1})), F(\cdot,\hat m_{\Cdot}(u_{2})) \ra_t \,\d u_{1} \d u_{2} \\
&=&
\int_0^T\!\!\!\int_0^T \p_2 F(t,\hat m_t(u_{1}))\p_2 F(t,\hat m_t(u_{2})) K(u_{1},t)K(u_{2},t) \d u_{1} \d u_{2} \,\d t. \\
\end{eqnarray*}
\end{proof}

\subsection{Exponential Volterra volatility model}\label{ssec:expVolterra}

Assume now that $X_t$ is the log-price process \eqref{log-model} with $\sigma_t$ being the \emph{exponential Volterra volatility process} 
\begin{equation}\label{e:expVolterra}
\sigma_t = g(t,Y_t) = \sigma_{0}\exp\left\{\xi Y_t \NEW{- \frac12 \alpha\xi^2 r(t)} \right\}, \quad t\geq 0,
\end{equation}
where $(Y_t,t\geq 0)$ is the Gaussian Volterra process \eqref{e:Volterra} satisfying assumptions \eqref{A1} and \eqref{A2}, $r(t)$ is its autocovariance function \eqref{e:r}, and $\sigma_0>0$, $\xi>0$ and $\alpha\in[0,1]$ are model parameters.

\begin{lemma}[Auxiliary terms in the decomposition formula for the exponential Volterra volatility model]\label{l:expVolterra}
Let $\sigma_t$ be as in \eqref{e:expVolterra} and $0\leq t\leq u$.
Then
\begin{align}
F(t,\hat{m}_{t}(u)) &= \sigma^{2}_{0}\exp \left\lbrace 2\xi \hat{m}_t(u) + 2\xi^2\hat{r}(u|t) \NEW{- \alpha\xi^2 r(u)}\right\rbrace,\\
\partial_2 F(t,\hat{m}_t(u)) &= 2\xi F(t,\hat{m}_t(u)),\\
\d\la M_{\Cdot},W_{\Cdot}\ra_{t}&= 2\sigma^{2}_{0}\xi \int^{T}_{0} \exp \left\{ 2\xi\hat m_t(u) + 2\xi^{2}\hat{r}(u|t) \NEW{-\alpha\xi^2 r(u)} \right\} K(u,t) \d u\, \d t,\\
d\la M_{\Cdot},M_{\Cdot}\ra_{t}&= 4\sigma^{4}_{0}\xi^2 \int_0^T \int_0^T \exp\left\{2\xi \left(\hat m_t(u_1)+\hat m_t(u_2)\right)\right\} \cdot\notag \\
&\qquad\qquad \cdot \exp\left\{2\xi^{2} \left(\hat{r}(u_1|t) + \hat{r}(u_2|t)\right)\right\} \cdot\notag \\
&\qquad\qquad \cdot \NEW{ \exp\left\{-\alpha\xi^2 \left( r(u_1) + r(u_2)\right) \right\} } \cdot\notag \\
&\qquad\qquad \cdot K(u_1,t)K(u_2,t) \d u_1 \d u_2 \, \d t.
\end{align} 
\end{lemma}
\begin{proof}
Let $\hat\varphi_t(u,z)$ be given by \eqref{eq:varphi}. Then
\begin{align*}
F(t,\hat m_t(u)) 
&=\int_{\mathbb{R}} g^2(u,z) \hat\varphi_t(u,z) \d z, \\
&=\sigma^{2}_{0} \NEW{\e^{-\alpha\xi^2 r(u)}} \int_{\mathbb{R}} \e^{2\xi z} \frac{1}{\sqrt{2\pi\hat{r}(u|t)}}\exp\left(-\frac{1}{2}\frac{(z-\hat m_t(u))^{2}}{\hat{r}(u|t)}\right) \d z.
\intertext{It is now easy to calculate the partial derivative $\p_2 F$. We get}
\partial_{2}F(t,\hat m_t(u)) &= \sigma^{2}_{0} \NEW{\e^{-\alpha\xi^2 r(u)}} \int_{\mathbb{R}} \e^{2\xi z} \frac{1}{\sqrt{2\pi\hat{r}(u|t)}}\exp\left(-\frac{1}{2}\frac{(z-\hat m_t(u))^{2}}{\hat{r}(u|t)}\right)\frac{z-\hat m_t(u)}{\hat{r}(u|t)} \d z.
\intertext{Changing variables $v=\frac{z-\hat m_t(u)}{\sqrt{\hat{r}(u|t)}}$ and $\d z=\sqrt{\hat{r}(u|t)} \d v$, we obtain}
\partial_{2}F(t,\hat m_t(u)) &= \frac{\sigma^{2}_{0} \NEW{\e^{-\alpha\xi^2 r(u)}} }{\sqrt{\hat{r}(u|t)}} e^{2\xi \hat m_t(u)} \int_{\mathbb{R}}  e^{2\xi \sqrt{\hat{r}(u|t)}v} v \phi(v) \d v\\
&= \frac{\sigma^{2}_{0} \NEW{\e^{-\alpha\xi^2 r(u)}} }{\sqrt{\hat{r}(u|t)}} e^{2\xi \hat m_t(u)} \E\left(e^{2\xi \sqrt{\hat{r}(u|t)}Z} Z \right),
\intertext{where $Z \sim \mathcal{N}(0,1)$. Using formula $\mathbb{E}[Ze^{\alpha Z}]=\alpha e^{\frac{\alpha^2}{2}}$, we get}
\partial_{2}F(t,\hat m_t(u)) &= 2\sigma^{2}_{0}\xi \exp \left\{ 2\xi\hat m_u(u) + 2\xi^{2}\hat{r}(u|t) \NEW{-\alpha\xi^2 r(u)} \right\} 
= 2\xi F(t,\hat{m}_t(u)).
\end{align*}
Remaining formulas follow accordingly.
\end{proof}

\begin{remark}
Using that $F(t,\hat{m}_{t}(u))=\E_{t}\left[\sigma^{2}_{u}\right]$, it is straightforward to see that

\begin{eqnarray}
\d M_{t} &=& 2\xi \left( \int^{T}_{t} \mathbb{E}_{t}\left[\sigma^{2}_{u}\right]K(u,t)\d u\right)\d W_{t} ,\label{e:dM}\\
\d\la M_{\Cdot},W_{\Cdot}\ra_{t} &=& 2\xi \int^T_0 \E_{t}\left[\sigma^{2}_{u}\right]K(u,t)\d u \d t,\\
\d\la M_{\Cdot},M_{\Cdot}\ra_{t} &=& 4\xi^2 \int^{T}_{0}\int^T_0 \E_{t}\left[\sigma^{2}_{u_{1}}\right]\E_{t}\left[\sigma^{2}_{u_{2}}\right]K(u_{1},t)K(u_{2},t) \d u_{1} \d u_{2}\d t.
\end{eqnarray}
\end{remark}

\begin{lemma} 
Let $\sigma_t$ be as in \eqref{e:expVolterra} and $0\leq t\leq u$. Then, we can re-write $F(t,\hat m_t(u))$ as
\begin{equation}\label{e:atu}
\E_{t}\left[\sigma^{2}_{u}\right]
=\sigma^{2}_{t}\exp\left\{-\alpha \xi^{2} (r(u)-r(t)) + 2\xi\int^{t}_{0}\left(K(u,z)-K(t,z)\right)dW_{z}+2\xi^{2}\hat{r}(u|t)\right\}.
\end{equation}

Moreover, we also have the following equalities
\begin{align*}
&\E_{t}\left[\sigma^{3}_{u}\exp\left\{2\xi\int^{u}_{0}\left(K(s,z)-K(u,z)\right)dW_{z}\right\}\right]\\
& =\sigma^{3}_{t}\exp\Bigl\{\NEW{-\frac{3}{2}\alpha \xi^{2}\left(r(u)-r(t)\right)}
+\xi\int^{t}_{0}\left(2K(s,z)+K(u,z)-3K(t,z)\right)dW_{z} \\
&\quad +\frac{\xi^{2}}{2}\int^{u}_{t}\left(2K(s,z)+K(u,z)\right)^{2}dz \Bigr\}
\end{align*}
and
\begin{align*}
&\E_{t}\left[\sigma^{4}_{u}\exp\left\{ 2\xi\int^{u}_{0}\left(K(s,z)+ K(v,z)-2K(u,z)\right)dW_{z} \right\}\right] \\
&=\sigma^{4}_{t}\exp\Bigl\{ \NEW{- 2\alpha\xi^{2}\left(r(u)-r(t)\right)}
+ 2\xi\int^{t}_{0}\left(K(s,z)+ K(v,z)-2K(t,z)\right)dW_{z} \\
&\quad + 2\xi^{2}\int^{u}_{t}\left(K(s,z)+ K(v,z)\right)^{2}dz \Bigr\}.
\end{align*}
\end{lemma}
\begin{proof}
The calculations to obtain these statements are straightforward.
\end{proof}

\begin{proposition}[Terms in the approximation formula for the exponential Volterra volatility model]\label{p:expVolterra}
Let $\sigma_t$ be as in \eqref{e:expVolterra} and $0\leq t\leq u$.
Then
\begin{align}
U_t 
&= \rho\xi\sigma^{3}_{t} \int_t^T \int^T_0  \exp\left\{\NEW{-\frac{3}{2}\alpha \xi^{2}\left(r(u)-r(t)\right)}+\xi\int^{t}_{0}\left(2K(s,z)+K(u,z)-3K(t,z)\right)\d W_{z}\right\} \notag \\
&\quad \cdot \exp\Bigl\{\frac{\xi^{2}}{2}\int^{u}_{t}\left(2K(s,z)+K(u,z)\right)^{2}\d z \NEW{-\alpha \xi^{2}\left(r(s)-r(u)\right)} + 2\xi^{2}\hat{r}(s|u)\Bigr\} \notag \\ 
&\quad \cdot K(s,u) \d s \d u
\intertext{and}
R_t 
&= \frac12 \xi^{2}\sigma^{4}_{t}\, \int^{T}_{t}\int^{T}_{0}\int^{T}_{0} \exp\Bigl\{ \NEW{-\alpha\xi^{2}\left(r(s)+r(v)-2r(t)\right)  } + 2\xi^{2}\left(\hat{r}(s|u) + \hat{r}(v|u)\right) \nonumber \\
&\quad + 2\xi\int^{t}_{0}\left(K(s,z)+ K(v,z)-2K(t,z)\right)\d W_{z}+ 2\xi^{2}\int^{u}_{t}\left(K(s,z)+ K(v,z)\right)^{2}\d z\Bigr\} \notag \\ 
&\quad \cdot K(s,u)K(v,u)\d s\d v \d u
\intertext{In particular,}
U_0 
&= \rho\xi\sigma^{3}_{0}\, \int_0^T \int_0^T \exp\Bigl\{\frac{\xi^2}{2} \int_0^u [2K(s,z)+K(u,z)]^2 \d z \Bigr\}\cdot\notag \\
& \quad \cdot \exp\Bigl\{ 2\xi^{2}\hat{r}(s|u) \NEW{-\frac12 \alpha\xi^2 r(u) - \alpha\xi^2 r(s)} \Bigr\} K(s,u)\,\d s \, \d u
\intertext{and}
R_0 
&= \frac12\sigma^{4}_{0}\xi^2\, \int_0^T \int_0^T \int_0^T \exp\Bigl\{2\xi^2 \int_0^u [K(s,z)+K(v,z)]^2 \d z \Bigr\} \cdot\notag \\
&\quad \cdot \exp\Bigl\{2\xi^{2} \left(\hat{r}(s|u) + \hat{r}(v|u)\right) \NEW{-\alpha\xi^2 \left( r(s) + r(v)\right)} \Bigr\} \cdot \notag \\
&\quad \cdot K(s,u)K(v,u) \d s \d v \, \d u.
\end{align}
\end{proposition}
\begin{proof}
We have that
\begin{align*}
U_t 
&= \frac{\rho}{2}\,   \E_t\left[\int_t^T \sigma_{u}\d\la M_{\Cdot},W_{\Cdot}\ra_{u}\right] \notag \\
&= \rho\xi\,   \E_t\left[\int_t^T \sigma_{u}\left(\int^T_0 \E_{u}\left[\sigma^{2}_{s}\right]K(s,u)\d s\right) \d u\right] \notag \\
&= \rho\xi\,  \int_t^T  \E_t\left[\sigma_{u}\left(\int^T_0 \E_{u}\left[\sigma^{2}_{s}\right]K(s,u)\d s\right)\right] \d u \notag \\
&= \rho\xi\,   \int_t^T\int^T_0   \E_t\left[\sigma^{3}_{u}\exp\Bigl\{\NEW{-\alpha \xi^{2}\left(r(s)-r(u)\right)} + 2\xi \int^{u}_{0}\left(K(s,z) -K(u,z)\right)\d W_{z} + 2\xi^{2}\hat{r}(s|u)\Bigr\}K(s,u)\right] \d s \d u \\
&= \rho\xi\,   \int_t^T \int^T_0  \E_t\left[\sigma^{3}_{u} \exp\Bigl\{2\xi \int^{u}_{0}\left(K(s,z) -K(u,z)\right)\d W_{z}\Bigr\}\right]\\
&\cdot \exp\Bigl\{\NEW{-\alpha \xi^{2}\left(r(s)-r(u)\right)} + 2\xi^{2}\hat{r}(s|u)\Bigr\}K(s,u) \d s \d u \\
&= \rho\xi\sigma^{3}_{t}\,   \int_t^T \int^T_0  \exp\left\{\NEW{-\frac{3}{2}\alpha \xi^{2}\left(r(u)-r(t)\right)}+\xi\int^{t}_{0}\left(2K(s,z)+K(u,z)-3K(t,z)\right)\d W_{z}\right\}\\
&\cdot \exp\Bigl\{\frac{\xi^{2}}{2}\int^{u}_{t}\left(2K(s,z)+K(u,z)\right)^{2}\d z \NEW{-\alpha \xi^{2}\left(r(s)-r(u)\right)} + 2\xi^{2}\hat{r}(s|u)\Bigr\}K(s,u) \d s \d u 
\end{align*}
Similarly, we have that 
\begin{align*}
R_t 
&= \frac18\,\E_t\left[\int^{T}_{t}\d\la M_{\Cdot},M_{\Cdot}\ra_{u}\right] \notag\\
&=\frac12 \xi^{2}\, \E_t\left[\int^{T}_{t} \left(\int^{T}_{0}\E_{u}\left[\sigma^{2}_{s}\right]K(s,u)\d s\right)^{2}\d u\right] \notag\\
&=\frac12 \xi^{2}\, \int^{T}_{t} \E_t\left[\left(\int^{T}_{0}\E_{u}\left[\sigma^{2}_{s}\right]K(s,u)\d s\right)^{2}\right]\d u \notag\\
&=\frac12 \xi^{2}\, \int^{T}_{t} \E_t\left[\left(\int^{T}_{0}\int^{T}_{0}\E_{u}\left[\sigma^{2}_{s}\right]\E_{u}\left[\sigma^{2}_{v}\right]K(s,u)K(v,u)\d s \d v\right)\right]\d u \notag\\
&=\frac12 \xi^{2}\, \int^{T}_{t}\E_t\left[\int^{T}_{0}\int^{T}_{0} \sigma^{4}_{u}K(s,u)K(v,u)
\exp\left\{\NEW{-\alpha \xi^{2} (r(s)+r(v)-2r(u))} \right. \right. \\
 &+ \left. \left.  2\xi\int^{u}_{0}\left(K(s,z)+ K(v,z)-2K(u,z)\right)\d W_{z} + 2\xi^{2}\left(\hat{r}(s|u) + \hat{r}(v|u)\right)\right\}\d s\d v\right] \d u \\
&=\frac12 \xi^{2}\, \int^{T}_{t}\int^{T}_{0}\int^{T}_{0} \E_t\left[\sigma^{4}_{u} \exp\left\{ 2\xi\int^{u}_{0}\left(K(s,z)+ K(v,z)-2K(u,z)\right)\d W_{z} \right\}\right]
\\
 & \exp\left\{\NEW{-\alpha \xi^{2} (r(s)+r(v)-2r(u))} + 2\xi^{2}\left(\hat{r}(s|u) + \hat{r}(v|u)\right)\right\}K(s,u)K(v,u)\d s\d v \d u \\
&=\frac12 \xi^{2}\sigma^{4}_{t}\, \int^{T}_{t}\int^{T}_{0}\int^{T}_{0} \exp\Bigl\{ \NEW{-\alpha\xi^{2}\left(r(s)+r(v)-2r(t)\right)  } + 2\xi^{2}\left(\hat{r}(s|u) + \hat{r}(v|u)\right) \\
 &+  2\xi\int^{t}_{0}\left(K(s,z)+ K(v,z)-2K(t,z)\right)\d W_{z}+ 2\xi^{2}\int^{u}_{t}\left(K(s,z)+ K(v,z)\right)^{2}\d z\Bigr\} \\ &\cdot K(s,u)K(v,u)\d s\d v \d u .
\end{align*}
\end{proof}

For the exponential Volterra volatility model we can determine an upper error bound for the price approximation in the following way.

\begin{theorem}[Upper error bound for the exponential Volterra volatility model]
\label{t:upper_error_bound}
Let $X_t$ be a log-price process \eqref{log-model} with $\sigma_t$ being the exponential Volterra volatility process \eqref{e:expVolterra}. Then we can express the call option fair value $V_t$ using the processes $R_t, U_t$ from Proposition \ref{p:expVolterra}. In particular,  
\begin{align*}
V_{t} &= BS(t,X_{t},v_{t})\\
&\quad +\Lambda\Gamma BS(t,X_{t},v_{t})U_{t}\\
&\quad +\Gamma^{2} BS(t,X_{t},v_{t})R_{t} \\
&\quad + \epsilon_{t},
\end{align*}
where $\epsilon_{t}$ are error terms of order $O\left(\left(\xi^{2} + \rho\xi\right)^{2}\right)$. 
\end{theorem}
\begin{proof}
Note that using \eqref{e:dM} we have that 
\begin{align*}
\d\left\langle M, M\right\rangle_{t}&=4\xi^{2}  \left(\int^{T}_{t} \mathbb{E}_{t}\left[\sigma^{2}_{u}\right] K(u,t)\d u\right)^{2} \d t, \\
\d\left\langle M, W\right\rangle_{t}&=2\xi \left(\int^{T}_{t} \mathbb{E}_{t}\left[\sigma^{2}_{u}\right] K(u,t)\d u\right) \d t.
\end{align*}

Applying the Jensen's inequality to \eqref{e:atu}, we can see that
\begin{align*}
a^{2}_{t}
&\geq\sigma^{2}_{t}(T-t) \exp\left\{\frac{1}{T-t}\int^{T}_{t} \left[ -\alpha \xi^{2} (r(u)-r(t)) + 2\xi\left(\hat{m}_{t}(u)-\hat{m}_{t}(t)\right) + 2\xi^{2}r(u|t) \right] \d u \right\}.
\intertext{Then, it is easy to find that}
\frac{T-t}{a^{2}_{t}} &\leq \frac1{\sigma^{2}_{0}} \exp\left\{-2\xi \hat{m}_{t}(t)+\alpha \xi^{2} r(t) - \frac{1}{T-t}\int^{T}_{t} \left[ -\alpha \xi^{2} (r(u)-r(t)) + 2\xi\left(\hat{m}_{t}(u)-\hat{m}_{t}(t)\right) + 2\xi^{2}r(u|t) \right] \d u\right\}
\end{align*}
where the exponent
\[
-2\xi \hat{m}_{t}(t)+\alpha \xi^{2} r(t) - \frac{1}{T-t}\int^{T}_{t}\left[-\alpha \xi^{2} (r(u)-r(t)) + 2\xi\left(\hat{m}_{t}(u)-\hat{m}_{t}(t)\right) + 2\xi^{2}r(u|t)\right] du.
\]
is a Gaussian process. Therefore $\frac{1}{a^{2}_{t}}$ has finite moments of all orders.

Using the error terms specified in the Appendix \ref{sec:error_terms} and Lemma \ref{lemaclau}, we find the following decompositions for each term

\begin{eqnarray}
\nonumber &&\left|\frac{1}{8}\mathbb{E}_{t}\left[\int^{T}_{t}e^{-r(u-t)}\Gamma^{2} BS(u,X_{u},v_{u})\d\left\langle M,M\right\rangle_{u}\right]-\Gamma^{2} BS(t,X_{t},v_{t})R_{t}\right|\\ \label{error_term_R}
&\leq&\frac{C}{8}\mathbb{E}_{t}\left[\int^{T}_{t}e^{-r(u-t)}\left(\frac{1}{a^7_{u}}+\frac{3}{a^6_{u}}+\frac{3}{a^5_{u}}+\frac{1}{a^4_{u}}\right)R_{u}\d\left\langle M,M\right\rangle_{u}\right]\\ \nonumber
&+&\frac{C\rho}{2}\mathbb{E}_{t}\left[\int^{T}_{t}e^{-r(u-t)}\left(\frac{1}{a^6_{u}}+\frac{2}{a^5_{u}}+\frac{1}{a^4_{u}}\right)R_{u}\sigma_{u}\d\left\langle W,M\right\rangle_{u}\right]\\ \nonumber
&+&C\rho\mathbb{E}_{t}\left[\int^{T}_{t}e^{-r(u-t)}\left(\frac{1}{a^4_{u}}+\frac{1}{a^3_{u}}\right)\sigma_{u} \d\left\langle W,R\right\rangle_{u}\right]\\ \nonumber
&+&\frac{C}{2}\mathbb{E}_{t}\left[\int^{T}_{t}e^{-r(u-t)}\left(\frac{1}{a^5_{u}}+\frac{2}{a^4_{u}}+\frac{1}{a^3_{u}}\right)\d\left\langle M,R\right\rangle_{u}\right].
\end{eqnarray}
and
\begin{eqnarray}
\nonumber &&\left|\frac{\rho}{2}\mathbb{E}_{t}\left[\int^{T}_{t}e^{-r(u-t)}\Lambda \Gamma BS(u,X_{u},v_{u})\sigma_{u}\d\left\langle W,M\right\rangle_{u}\right]-\Lambda \Gamma BS(t,X_{t},v_{t})U_{t}\right|\\ \label{error_term_U}
&\leq&\frac{C}{8}\mathbb{E}_{t}\left[\int^{T}_{t}e^{-r(u-t)}\left(\frac{1}{a^6_{u}}+\frac{2}{a^5_{u}}+\frac{1}{a^4_{u}}\right)U_{u}\d\left\langle M,M\right\rangle_{u}\right]\\ \nonumber
&+&\frac{C\rho}{2}\mathbb{E}_{t}\left[\int^{T}_{t}e^{-r(u-t)}\left(\frac{1}{a^5_{u}}+\frac{1}{a^4_{u}}\right)U_{u}\sigma_{u}\d\left\langle W,M\right\rangle_{u}\right]\\ \nonumber
&+&C\rho\mathbb{E}_{t}\left[\int^{T}_{t}e^{-r(u-t)}\frac{1}{a^3_{u}}\sigma_{u} \d\left\langle W,U\right\rangle_{u}\right]\\ \nonumber
&+&\frac{C}{2}\mathbb{E}_{t}\left[\int^{T}_{t}e^{-r(u-t)}\left(\frac{1}{a^4_{u}}+\frac{1}{a^3_{u}}\right)\d\left\langle M,U\right\rangle_{u}\right].
\end{eqnarray}
Since all previous conditional expectations are obviously finite, the error order is $O\left(\left(\xi^{2} + \rho\xi\right)^{2}\right)$.

\end{proof}

Further, we express differentials with respect to the $n^{\text{th}}$-power of the exponential Volterra volatility process when $Y_{t}$ is a semimartingale.

\begin{lemma}\label{l:d_sigma} 
Let $\sigma_t$ be as in \eqref{e:expVolterra} and $Y_{t}$ a semimartingale. Let $n\geq 1$, we have that
\begin{eqnarray}
\d\sigma^{n}_{t}&=&\sigma^{n}_{t}K(t,t)\left[n\xi \d W_{t} + \frac{n}{2} \xi^{2}K(t,t)\left(n-\alpha\right) \d t\right].
\end{eqnarray}
\end{lemma}
\begin{proof}
The formula is an immediate consequence of the It\^o formula.
\end{proof}

\begin{lemma}\label{l:covariation}
Let $\sigma_t$ be as in \eqref{e:expVolterra} and $Y_{t}$ is a semimartingale. We can calculate $\d U_{t}$ and $\d R_{t}$. In order to simplify the notation, we define the
two following functions
\begin{align*}
\varphi(t,s,x,T) &:=\exp\left\{-\frac{3}{2}\alpha \xi^{2}\left(r(x)-r(u)\right)+\xi\int^{u}_{0}\left(2K(s,z)+K(x,z)-3K(u,z)\right)\d W_{z}\right\}\cdot \\
&\quad\cdot \exp\Bigl\{\frac{\xi^{2}}{2}\int^{x}_{u}\left(2K(s,z)+K(x,z)\right)^{2}\d z -\alpha \xi^{2}\left(r(s)-r(x)\right) + 2\xi^{2}\hat{r}(s|x)\Bigr\},\\
\psi(t,s,v,x,T) &:=\exp\Bigl\{ -\alpha\xi^{2}\left(r(s)+r(v)-2r(t)\right)   + 2\xi^{2}\left(\hat{r}(s|x) + \hat{r}(v|x)\right) \\  
&\quad +2\xi\int^{t}_{0}\left(K(s,z)+ K(v,z)-2K(t,z)\right)\d W_{z}+ 2\xi^{2}\int^{x}_{t}\left(K(s,z)+ K(v,z)\right)^{2}\d z\Bigr\}.
\intertext{Then }
&\d U_t = \rho\xi \sigma^{3}_{t}K(t,t)\left[3\xi \d W_{t} + \frac{3}{2} \xi^{2}K(t,t)\left(3-\alpha\right) \d t\right]\,   \int_t^T \int^T_0  \varphi(t,s,x,T) K(s,x) \d s \d x\\
&- \rho\xi \sigma^{3}_{t} \int^T_0  \varphi(t,s,t,T) K(s,x) \d s \d t\\
&+\rho\xi  \sigma^{3}_{t}\,   \int_t^T \int^T_0  \varphi(t,s,x,T)\Bigl\{ \frac{3}{2}\alpha \xi^{2} \d r(t) + \xi \left(2K(s,t)+K(x,t)-3K(t,t)\right)\d W_{t} \\
&\quad +\frac12\xi^{2} \left(2K(s,t)+K(x,t)-3K(t,t)\right)^{2}\d t - \frac{\xi^{2}}{2}\left(2K(s,t)+K(x,t)\right)^{2}\d t\Bigr\} K(s,x) \d s \d x.\\
\intertext{and }
&\d R_t = \frac12 \xi^{2} \sigma^{4}_{t}K(t,t)\left[4\xi \d W_{t} + 2 \xi^{2}K(t,t)\left(4-\alpha\right) \d t\right]\\
&\, \int^{T}_{t}\int^{T}_{0}\int^{T}_{0}\Phi(t,s,v,x,T)\cdot K(s,x)K(v,x) \d s\d v \d x\\
&-\frac12 \xi^{2} \sigma^{4}_{t}\int^{T}_{0}\int^{T}_{0}\Phi(t,s,v,t,T) K(s,t)K(v,t) \d s\d v \d t\\
&+\frac12 \xi^{2} \sigma^{4}_{t}\, \int^{T}_{t}\int^{T}_{0}\int^{T}_{0}  \Phi(t,s,v,x,T)K(s,x)K(v,x)\\
&\quad \Bigl\{ 2\alpha\xi^{2}\d r(t) + 2\xi\left(K(s,t)+ K(v,t)-2K(t,t)\right)\d W_{t}  \\
&\quad +  2\xi^{2} \left(K(s,t)+ K(v,t)-2K(t,t)\right)^{2} dt - 2\xi^{2}\left(K(s,t)+ K(v,t)\right)^{2}\d t\Bigr\} \d s\d v \d x\\
\intertext{We define the following auxiliary function }
&\zeta(t,T) := \int^{T}_{t} \mathbb{E}_{t}\left[\sigma^{2}_{z}\right]K(z,t)\d z. \\
\intertext{Then, it is easier to see that the covariations are the following}
\d\la U,W\ra_t &= 3\rho\xi^{2} \sigma^{3}_{t} K(t,t)  \int_t^T \int^T_0  \varphi(t,s,x,T) K(s,x) \d s \d x\,\d t\\
&\quad +\rho\xi^{2}  \sigma^{3}_{t} \int_t^T \int^T_0  \varphi(t,s,x,T) \left(2K(s,t)+K(x,t)-3K(t,t)\right)\, K(s,x) \d s \d x\, \d t,\\
\d\la U,M\ra_t &=6\rho\xi^{3} \sigma^{3}_{t}K(t,t) \zeta(t,T) \int_t^T \int^T_0  \varphi(t,s,x,T) K(s,x) \d s \d x\,\d t\\
&+2\rho\xi^{3}  \sigma^{3}_{t} \zeta(t,T) \int_t^T \int^T_0  \varphi(t,s,x,T) \left(2K(s,t)+K(x,t)-3K(t,t)\right) K(s,x) \d s \d x\,\d t,\\
\d\la R, W\ra_{t} &=2 \xi^{3} \sigma^{4}_{t}K(t,t) \, \int^{T}_{t}\int^{T}_{0}\int^{T}_{0}\psi(t,s,v,x,T)\cdot K(s,x)K(v,x) \d s\d v \d x\,\d t\\
&+\xi^{3} \sigma^{4}_{u}  \, \int^{T}_{u}\int^{T}_{0}\int^{T}_{0}  \psi(t,s,v,x,T)K(s,x)K(v,x)\left(K(s,u)+ K(v,u)-2K(u,u)\right) \d s\d v \d x\,\d t\\
\intertext{and}
\d\la R, M\ra_{t} &=4 \xi^{4} \sigma^{4}_{t}K(t,t) \zeta(t,T) \int^{T}_{t}\int^{T}_{0}\int^{T}_{0}\psi(t,s,v,x,T)\cdot K(s,x)K(v,x) \d s\d v \d x\,\d t\\
&+2\xi^{4} \sigma^{4}_{t}  \zeta(t,T) \int^{T}_{t}\int^{T}_{0}\int^{T}_{0}  \psi(t,s,v,x,T)K(s,x)K(v,x)\left(K(s,t)+ K(v,t)-2K(t,t)\right) \d s\d v \d x\, \d t.
\end{align*}
\end{lemma}
\begin{proof}
\begin{align*}
\intertext{Now, we can re-write $U_{t}$ as }
&U_t=\rho\xi\sigma^{3}_{t}\,   \int_t^T \int^T_0  \varphi(t,s,x,T) K(s,x) \d s \d x
\intertext{and $R_{t}$ as}
&R_t = \frac12 \xi^{2}\sigma^{4}_{t}\, \int^{T}_{t}\int^{T}_{0}\int^{T}_{0}\psi(t,s,v,x,T)\cdot K(s,x)K(v,x) \d s\d v \d x.
\intertext{We have that}
&\d U_t = \rho\xi \d \sigma^{3}_{t}\,   \int_t^T \int^T_0  \varphi(t,s,x,T) K(s,x) \d s \d x\\
&- \rho\xi \sigma^{3}_{t} \int^T_0  \varphi(t,s,t,T) K(s,x) \d s \d t\\
&+\rho\xi \sigma^{3}_{t}\,   \int_t^T \int^T_0  \varphi(t,s,x,T) \Bigl\{ \frac{3}{2}\alpha \xi^{2} \d r(t) + \xi \left(2K(s,t)+K(x,t)-3K(t,t)\right)\d W_{t}  \\
&\quad + \xi^{2} \left(2K(s,t)+K(x,t)-3K(t,t)\right)^{2}\d t - \frac{\xi^{2}}{2}\left(2K(s,t)+K(x,t)\right)^{2}\d t\Bigr\} K(s,x) \d s \d x.\\
\intertext{Using Lemma \ref{l:d_sigma}, we obtain}
&\d U_t = \rho\xi \sigma^{3}_{t}K(t,t)\left[3\xi \d W_{t} + \frac{3}{2} \xi^{2}K(t,t)\left(3-\alpha\right) \d t\right]\,   \int_t^T \int^T_0  \varphi(t,s,x,T) K(s,x) \d s \d x\\
&- \rho\xi \sigma^{3}_{t} \int^T_0  \varphi(t,s,t,T) K(s,x) \d s \d t\\
&+\rho\xi  \sigma^{3}_{t}\,   \int_t^T \int^T_0  \varphi(t,s,x,T)\Bigl\{ \frac{3}{2}\alpha \xi^{2} \d r(t) + \xi \left(2K(s,t)+K(x,t)-3K(t,t)\right)\d W_{t} \\
&\quad +\frac12\xi^{2} \left(2K(s,t)+K(x,t)-3K(t,t)\right)^{2}\d t - \frac{\xi^{2}}{2}\left(2K(s,t)+K(x,t)\right)^{2}\d t\Bigr\} K(s,x) \d s \d x.\\
\intertext{We have that}
&\d R_t = \frac12 \xi^{2} \d \sigma^{4}_{t}\, \int^{T}_{t}\int^{T}_{0}\int^{T}_{0}\Phi(t,s,v,x,T)\cdot K(s,x)K(v,x) \d s\d v \d x\\
&-\frac12 \xi^{2} \sigma^{4}_{t}\int^{T}_{0}\int^{T}_{0}\Phi(t,s,v,t,T) K(s,t)K(v,t) \d s\d v \d t\\
&+\frac12 \xi^{2} \sigma^{4}_{t}\, \int^{T}_{t}\int^{T}_{0}\int^{T}_{0}  \Phi(t,s,v,x,T)K(s,x)K(v,x)\\
&\quad \Bigl\{ 2\alpha\xi^{2}\d r(t) + 2\xi\left(K(s,t)+ K(v,t)-2K(t,t)\right)\d W_{t}  \\
&\quad + 2\xi^{2} \left(K(s,t)+ K(v,t)-2K(t,t)\right)^{2} dt - 2\xi^{2}\left(K(s,t)+ K(v,t)\right)^{2}\d t\Bigr\}\d s\d v \d x.\\
\intertext{Using Lemma \ref{l:d_sigma}, we obtain}
&\d R_t = \frac12 \xi^{2} \sigma^{4}_{t}K(t,t)\left[4\xi \d W_{t} + 2 \xi^{2}K(t,t)\left(4-\alpha\right) \d t\right]\\
&\, \int^{T}_{t}\int^{T}_{0}\int^{T}_{0}\Phi(t,s,v,x,T)\cdot K(s,x)K(v,x) \d s\d v \d x\\
&-\frac12 \xi^{2} \sigma^{4}_{t}\int^{T}_{0}\int^{T}_{0}\Phi(t,s,v,t,T) K(s,t)K(v,t) \d s\d v \d t\\
&+\frac12 \xi^{2} \sigma^{4}_{t}\, \int^{T}_{t}\int^{T}_{0}\int^{T}_{0}  \Phi(t,s,v,x,T)K(s,x)K(v,x)\\
&\quad \Bigl\{ 2\alpha\xi^{2}\d r(t) + 2\xi\left(K(s,t)+ K(v,t)-2K(t,t)\right)\d W_{t}  \\
&\quad +  2\xi^{2} \left(K(s,t)+ K(v,t)-2K(t,t)\right)^{2} dt - 2\xi^{2}\left(K(s,t)+ K(v,t)\right)^{2}\d t\Bigr\} \d s\d v \d x.\\
\end{align*}
\end{proof}

\begin{remark}
Note that the Theorem \ref{t:upper_error_bound} requires that the conditional expectations of (\ref{error_term_R}) and (\ref{error_term_U}) be finite. In the case where $\sigma_t$ is as in \eqref{e:expVolterra} and $Y_{t}$ is a semimartingale, it easy to see that these conditions are met using Lemma \ref{l:covariation}.
\end{remark}

\subsection{Exponential fractional volatility model}\label{ssec:fBm}

Let us now focus on a very important example of Gaussian Volterra processes, the \emph{fractional Brownian motion} (fBm), which is a process with a Hurst parameter $H\in(0,1)$, with covariance function 
\begin{equation}\label{e:fBm_cov}
r(t,s):=\E[B^H_t B^H_s] = \frac12 \left( t^{2H} + s^{2H} - |t-s|^{2H}\right), \quad t,s\geq 0,
\end{equation} 
and, in particular, with variance
\begin{equation}\label{e:fBm_acov}
r(t):=r(t,t) = t^{2H},\quad t\geq0. 
\end{equation}

One of the first applications of fractional Brownian motion is credited to \cite{Hurst51} who modelled the long term storage capacity of reservoirs along the Nile river. However, the idea of this concept goes back to \cite{Kolmogorov40}, who studied Wiener spirals and some other curves in Hilbert spaces. Later, \cite{Levy53} used the Riemann--Liouville fractional integral to define the process as
\[ 
\tilde{B}^H_t := \frac{1}{\Gamma(H+1/2)} \int\limits_0^t (t-s)^{H-1/2} \d W_s,
\] 
where $H$ may be any positive number. This type of integral turned out to be ill-suited to applications of fractional Brownian motion because of its over-emphasis on the origin for many applications. 
In their highly cited work, \cite{Mandelbrot68} introduced the Weyl's representation of the fractional Brownian motion:    
\begin{equation}\label{e:sfBm}
B^H_{t}:=\frac{1}{\Gamma(H+1/2)} \left[ Z_t + \int\limits_0^t (t-s)^{H-1/2} \d W_s \right],
\end{equation}
where 
$$Z_t:=\int\limits_{-\infty}^0\left[(t-s)^{H-1/2}-(-s)^{H-1/2}\right] \d W_s$$
and $W_t$ is the standard Wiener process. Nowadays, the most widely used representation of fBm is the one by \cite{Molchan69} 
\begin{align}
B^H_t &:= \int\limits_0^t K_H(t,s) \d W_s, \label{e:Molchan_fBm}
\intertext{where}
K_H(t,s) &:= C_H \Biggl[ \left( \frac{t}{s} \right)^{H-\frac12} (t-s)^{H-\frac12} 
- \left(H-\frac12\right) s^{H-\frac12} \int_s^t z^{H-\frac32} (z-s)^{H-\frac12} \d z \Biggr] \label{e:Molchan_kernel} \\
C_H &:= \sqrt{\frac{2 H \Gamma\left(\frac32-H\right)}{\Gamma\left(H+\frac12\right)\Gamma\left(2-2H\right)} }. \notag
\end{align}
To understand the connection between Molchan-Golosov and Mandelbrot-Van Ness representations of fBm we refer readers to the paper by \cite{Jost08}.

Despite of the above mentioned arguments, \cite{Alos00} proposed to consider a process $\hat B_t = \int_0^t (t-s)^{H-1/2} \d W_s$ instead of $B^H_t$ in fractional stochastic calculus, since $Z_t$ has absolutely continuous trajectories. Since $B^H_t$ is not a semimartingale, the process $\hat B_t = \Gamma(H+1/2)B^H_t - Z_t$ is also not a semimartingale. Later on, \cite{Thao06} introduced the so called approximate fractional Brownian motion process as 
\begin{equation*}
\hat B^\varepsilon_{t} = \int\limits_0^t (t-s+\varepsilon)^{H-1/2} \d W_s, \quad H\in(0,1), H\ne\frac12, \varepsilon>0,
\end{equation*}
and showed that for every $\varepsilon>0$ the process $\hat B^\varepsilon_t$ is a semimartingale and it converges to $\hat B_t$ in $L^2(\Omega)$ when $\varepsilon$ tends to zero. This convergence is uniform with respect to $t\in[0,T]$ \citep[Theorem 2.1]{Thao06}.

Let us now consider the \emph{exponential fractional volatility process} 
\begin{equation}\label{e:aRFSV}
\sigma_t := \sigma_{0}\exp\left\{\xi B^H_t \NEW{- \frac12 \alpha\xi^2 r(t)} \right\}, \quad t\geq 0,
\end{equation}
where $(B^H_t,t\geq 0)$ is one of the above mentioned representations of fBm. We are especially interested in the "rough" models, i.e. when $H<1/2$. In this case, we call the model \eqref{model} with volatility process \eqref{e:aRFSV} the \emph{rough fractional stochastic volatility model} ($\alpha$RFSV). Note that if $\alpha=1$, we get the rBergomi model \citep{Bayer16,Gatheral18}, if $\alpha=0$, we get the original exponential fractional volatility model. Values of $\alpha$ between 0 and 1 give us a new degree of freedom that can be viewed as a weight between these two models.

\begin{example}[Volatility driven by the approximate fractional Brownian motion]\label{ex:afbm} 
Let us consider model \eqref{log-model} with volatility process 
\begin{align*}
\sigma_t &= \sigma_{0} \exp\left\{\xi \tilde B_t \NEW{- \frac12 \alpha\xi^2 r(t)} \right\}, 
\intertext{where}
\tilde B_t &= \int_0^t \tilde K(t,s) \d W_s 
\intertext{and}
\tilde K(t,s) &= \sqrt{2H}(t-s+\varepsilon)^{H-1/2}, \quad s\leq t, \varepsilon\geq 0, H\in(0,1). \notag
\intertext{Then}
r(t,s) &= \int_0^{t\wedge s} \tilde K(t,v) \tilde K(s,v) \d v, \\
r(t)   &= \int_0^t \tilde K^2(t,v) \d v = 2H \int_0^t (t-v+\varepsilon)^{2H-1} \d v = (t+\varepsilon)^{2H} - \varepsilon^{2H}.
\intertext{Note that if $\varepsilon=0$, we get exactly the variance $r(t)=t^{2H}$ , that it is the variance of the standard fractional Brownian motion. Further we have}
\hat r(t|u) &= r(t) - \int_0^u \tilde K^2(t,v) \d v = r(t) - 2H \int_0^u (t-v+\varepsilon)^{2H-1} \d v = (t-u+\varepsilon)^{2H} - \varepsilon^{2H}
\intertext{and thus}
U_0 &= \rho \sigma^{3}_{0}\xi\sqrt{2H} \int_0^T \int_0^T \exp\Bigl\{ \xi^2 H \int_0^u [(u-v+\varepsilon)^{H-1/2} + 2(s-v+\varepsilon)^{H-1/2}]^2 \d v \Bigr\} \cdot \\
&\qquad\qquad \cdot \exp\Bigl\{ 2\xi^2 [(s-u+\varepsilon)^{2H} - \varepsilon^{2H}] \Bigr\} \cdot \\
&\qquad\qquad \cdot \NEW{ \exp\Bigl\{ -\frac12 \alpha\xi^2 [(u+\varepsilon)^{2H} + 2(s+\varepsilon)^{2H} - 3\varepsilon^{2H}] \Bigr\} \cdot } \\
&\qquad\qquad \cdot (s-u+\varepsilon)^{H-1/2} \d s \d u \\
R_0 &=   \sigma^{4}_{0}\xi^2 H \int_0^T \int_0^T \int_0^T \exp\Bigl\{ 4 \xi^2 H \int_0^u [(t_1-v+\varepsilon)^{H-1/2} + (t_2-v+\varepsilon)^{H-1/2}]^2 \d v \Bigr\} \cdot \\
&\qquad\qquad \cdot \exp\Bigl\{ 2 \xi^2 [(t_1-u+\varepsilon)^{2H} + (t_2-u+\varepsilon)^{2H}-2\varepsilon^{2H}] \Bigr\} \cdot \\
&\qquad\qquad \cdot \NEW{ \exp\Bigl\{ - \alpha \xi^2 [(t_1+\varepsilon)^{2H} + (t_2+\varepsilon)^{2H}-2\varepsilon^{2H}] \Bigr\} \cdot} \\
&\qquad\qquad \cdot (t_1-u+\varepsilon)^{H-1/2} (t_2-u+\varepsilon)^{H-1/2} \d t_1 \d t_2 \d u.
\end{align*}
\end{example}

\begin{example}[Volatility driven by the standard Wiener process]\label{ex:ExpWiener}
If in the previous Example~\ref{ex:afbm} we take $H=1/2$ and $\varepsilon=0$, we get model \eqref{log-model} with exponential Wiener volatility process 
\begin{align}
\sigma_t &= \sigma_0\exp\left\{\xi \tilde W_t \NEW{- \frac12 \alpha \xi^2 r(t)} \right\}, \label{e:Wienervol} 
\intertext{where}
\tilde W_t &= \int_0^t \tilde K(t,s) \d W_s \notag
\intertext{is the standard Wiener process, i.e. where $\tilde K(t,s) = \1_{\{s\leq t\}}$. In this case, we have that }
v^{2}_{t}&=\frac{\sigma^{2}_t}{ \left(2-\alpha\right) \xi^2 (T-t)} \left[ \exp\{\left(2-\alpha\right)\xi^2 (T-t)\} -1\right],\notag\\
r(t,s) &= \int_0^{t\wedge s} \tilde K(t,v) \tilde K(s,v) \d v = t\wedge s, \notag\\
r(t)   &= \int_0^t \tilde K^2(t,v) \d v = t.\notag
\intertext{Define}
\phi(t,T\NEW{,\alpha})&:=\int^{T}_{t} \exp\left\{(2\NEW{-\alpha})\xi^{2}  (s-t)\right\} \d s.\label{e:phi} \\
\intertext{It is easy to see that}
\d M_{t} &= 2\xi \sigma^{2}_{t} \d W_{t} \phi(t,T\NEW{,\alpha}),\notag
\intertext{and thus}
U_0 &=  \rho \sigma^{3}_{0}\xi \int_0^T \int_0^T \exp\Bigl\{ \frac12 \xi^2 \int_0^u [\1_{\{v\leq u\}} + 2\cdot\1_{\{v\leq s\}}]^2 \d v \Bigr\} \cdot \notag\\
&\qquad\qquad \cdot \exp\Bigl\{ 2\xi^2 (s-u) \NEW{-\frac12 \alpha \xi^2 u - \alpha \xi^2 s} \Bigr\} \1_{\{u\leq s\}} \d s \d u \notag\\
&= \rho\sigma^{3}_{0}\xi \int_0^T \int_0^T \exp\Bigl\{ \frac92 \xi^2 u \Bigr\} \exp\Bigl\{ \frac12 \xi^2 [(4\NEW{-2\alpha})s - (4\NEW{+\alpha})u ] \Bigr\} \1_{\{u\leq s\}} \d s \d u \\
&= \frac{2\rho\sigma^{3}_{0}}{3(2\NEW{-\alpha})(3\NEW{-\alpha})(5\NEW{-\alpha})\xi^{3}} \left[2(2\NEW{-\alpha})\exp\Bigl\{ \frac32 \xi^2 (3\NEW{-\alpha}) T \Bigr\} -3(3\NEW{-\alpha})\exp\Bigl\{ \xi^2 (2\NEW{-\alpha})T \Bigr\}+5\NEW{-\alpha}\right] \label{e:WienerU0}
\intertext{and}
R_0 
&= \frac12 \sigma^{4}_{0}\xi^2 \int_0^T \int_0^T \int_0^T \exp\Bigl\{ 2 \xi^2 \int_0^u [ \1_{\{v\leq t_1\}} + \1_{\{v\leq t_2\}} ]^2 \d v \Bigr\} \cdot \notag\\
&\qquad\qquad \cdot \exp\Bigl\{ 2\xi^2 [ t_1 + t_2 - 2u ] \NEW{-\alpha\xi^2[t_1+t_2]} \Bigr\} \1_{\{u\leq t_1\}} \1_{\{u\leq t_2\}} \d t_1 \d t_2 \d u \notag\\
&= \frac12 \sigma^{4}_{0}\xi^2 \int_0^T \int_0^T \int_0^T \exp\Bigl\{ 8 \xi^2 u\Bigr\} \exp\Bigl\{ 2\xi^2 [ t_1 + t_2 - 2u ]\NEW{-\alpha\xi^2[t_1 + t_2]} \Bigr\} \1_{\{u\leq t_1\}} \1_{\{u\leq t_2\}} \d t_1 \d t_2 \d u \\
&=\frac{\sigma^{4}_{0}}{8(2\NEW{-\alpha})^2(4\NEW{-\alpha})(6\NEW{-\alpha})\xi^{4}} \Bigl[ 
(2\NEW{-\alpha})^2 \exp\Bigl\{2(4\NEW{-\alpha})\xi^2 T\Bigr\} \notag\\
&\qquad\qquad - (4\NEW{-\alpha})(6\NEW{-\alpha}) \exp\Bigl\{2(2\NEW{-\alpha})\xi^2 T\Bigr\} \notag\\
&\qquad\qquad + 8(4\NEW{-\alpha}) \exp\Bigl\{(2\NEW{-\alpha})\xi^2 T\Bigr\} \label{e:WienerR0}
- 2(6\NEW{-\alpha}) \Bigr].
\intertext{For a model without exponential drift ($\alpha=0$) these formulas simplify to}
U_0 &= \frac{\rho\sigma^{3}_{0}}{45\xi^{3}}  \left[4\exp\Bigl\{ \frac92 \xi^2 T \Bigr\} -9\exp\Bigl\{ 2 \xi^2 T \Bigr\}+5\right] \notag\\
R_0 &=\frac{\sigma^{4}_{0}}{192\xi^{4}} \left[\exp\Bigl\{2 \xi^{2}T\Bigr\} -1 \right]^3 \left[\exp\Bigl\{ 2 \xi^2 T \Bigr\} + 3 \right] \notag
\intertext{and for the classical Bergomi model ($\alpha=1$) we get}
U_0 &= \frac{\rho\sigma^{3}_{0}}{6\xi^{3}}  \left[\exp\Bigl\{ 3\xi^2 T \Bigr\} -3\exp\Bigl\{ \xi^2 T \Bigr\}+2\right] \notag\\
R_0 &=\frac{\sigma^{4}_{0}}{120\xi^{4}} \left[\exp\Bigl\{6 \xi^{2}T\Bigr\} -15\exp\Bigl\{ 2 \xi^2 T \Bigr\} + 24\exp\Bigl\{ \xi^2 T \Bigr\} -10\right]. \notag
\intertext{For matter of convenience, we define the functions}
\psi(t,T\NEW{,\alpha}) &= \int^{T}_{t}\exp\left\{(8\NEW{-2\alpha})\xi^{2} (s-t) \right\} \left[\exp\left\{(2\NEW{-\alpha})\xi^{2}  (T-s)\right\}-1\right]^{2} \d s
\intertext{and}
\zeta(t,T\NEW{,\alpha}) &= \int^{T}_{t} \exp\left\{\frac{1}{2}(9-3\NEW{\alpha})\xi^{2} (s-t) \right\} \left[\exp\left\{(2\NEW{-\alpha})\xi^{2}  (T-s)\right\}-1\right] \d s.
\intertext{We can re-write $U_{t}$ and $R_{t}$ as}
U_{t} &= \frac{\rho\sigma^{3}_{t}}{(2\NEW{-\alpha})\xi}\zeta(t,T\NEW{,\alpha})
\intertext{and }
R_{t} &= \frac{\sigma^{4}_{t}}{2(2\NEW{-\alpha})^{2}\xi^{2}}\psi(t,T\NEW{,\alpha}).
\intertext{It is easy to find the $\d U_{t}$ and $\d R_{t}$,}
\d U_{t} 
&= \frac{\rho \d\sigma^{3}_{t}}{(2\NEW{-\alpha})\xi}\zeta(t,T\NEW{,\alpha}) + \frac{\rho\sigma^{3}_{t}}{(2\NEW{-\alpha})\xi}\zeta ' (t,T\NEW{,\alpha}) \d t \\
&= \frac{\rho \left(3\xi\sigma^{3}_{t} \d W_{t} + \frac{1}{2}(18-3\NEW{\alpha})\xi^{2} \sigma^{3} \d t\right)}{(2\NEW{-\alpha})\xi}\zeta(t,T\NEW{,\alpha}) + \frac{\rho\sigma^{3}_{t}}{(2\NEW{-\alpha})\xi}\zeta ' (t,T\NEW{,\alpha}) \d t \\
\intertext{and}
\d R_{t} 
&= \frac{\d\sigma^{4}_{t}}{2(2\NEW{-\alpha})^{2}\xi^{2}}\psi(t,T\NEW{,\alpha}) + \frac{\sigma^{4}_{t}}{2(2\NEW{-\alpha})^{2}\xi^{2}}\psi ' (t,T\NEW{,\alpha}) \d t \\
&= \frac{4\xi\sigma^{4}_{t}\d W_{t} + 2(8\NEW{-\alpha})\sigma^{4}_{t} \d t}{2(2\NEW{-\alpha})^{2}\xi^{2}}\psi(t,T\NEW{,\alpha}) + \frac{\sigma^{4}_{t}}{2(2\NEW{-\alpha})^{2}\xi^{2}}\psi ' (t,T\NEW{,\alpha}) \d t.
\end{align}

\end{example}
\begin{remark}
We can do a Taylor expansion of $U_{0}$ and $R_{0}$ to understand better their dependencies, doing that we obtain
\begin{align*}
U_{0}&\sim\rho \xi   T^2 \sigma^{3}_{t} \Biggl(\frac{1}{2} +\frac{1}{12} (13-5\NEW{\alpha} ) \xi^2  T+\frac{1}{96}   (\NEW{\alpha}  (19 \NEW{\alpha} -100)+133) \xi^4 T^2 \\
 &\qquad -\frac{1}{960}  \left(5 \NEW{\alpha} -13\right) \left(\NEW{\alpha} (13 \NEW{\alpha} -70)+97 \right)\xi^6 T^3+O\left(T^4\right)\Biggr)\\
\intertext{and }
R_{0}&\sim  \xi^2  T^3 \sigma^4_{t}\left(-\frac{1}{6}  \left(\NEW{\alpha}-2\right) + \frac{1}{24} (\NEW{\alpha} -2) (5\NEW{\alpha} -14) \xi^2 T\right. \\
 &- \left.\frac{1}{120}  \left(\NEW{\alpha} -2\right) \left(17 \NEW{\alpha}^2-96\NEW{\alpha} +140\right)\xi^4T^2 + O\left(T^3\right)\right).
\end{align*}
\end{remark}

\begin{theorem}[Decomposition formula for exponential Wiener volatility model]
\label{t:expWiener}
Let $X_t$ be the log-price process \eqref{log-model} with $\sigma_t$ being the exponential Wiener volatility process defined in \eqref{e:Wienervol}. Assuming without any loss of generality that the options starts at time 0, then we can express the call option fair value $V_0$ using the processes $U_0, R_0$ from \eqref{e:WienerU0} and \eqref{e:WienerR0} respectively. In particular,
\begin{align*}
V_{0} &= BS(0,X_{0},v_{0})\\
&\quad +\Lambda\Gamma BS(0,X_{0},v_{0})U_{0}\\
&\quad +\Gamma^{2} BS(0,X_{0},v_{0})R_{0} \\
&\quad + \epsilon.
\end{align*}
where $\epsilon$ denotes error terms and for $\alpha\geq0$, $\left|\epsilon\right|$ is at most of the order $C\xi (\sqrt{T}+\rho \xi^{2})T^{3/2}\Pi(\alpha,T, \xi, \rho).$ The exact bound is given in Appendix \ref{sec:Upp_bound_Wiener}.
\end{theorem}
\begin{proof}
The detailed proof is given in Appendix \ref{sec:Upp_bound_Wiener}, where we also examine the order of magnitude by the first Taylor term of the integrals.
\end{proof}

\begin{remark}
It is worth to mention that the order of the error bound from Theorem~\ref{t:expWiener} is better than the general estimate from Theorem~\ref{t:upper_error_bound}, where the time dependency is not considered. To get finer estimates also for the exponential fractional model (case $H\ne 1/2$), a proof similar to the one in Appendix \ref{sec:Upp_bound_Wiener} would have to be performed with more complicated but still tractable calculations.
\end{remark}

\begin{example}[Volatility driven by the standard fractional Brownian motion] Let us consider a model with volatility process 
$$\sigma_t = \exp\{\xi B^H_t \NEW{-\frac12\alpha\xi^2 r(t) } \},$$ 
where $B^H_t$ is the standard fractional Brownian motion as defined in \eqref{e:Molchan_fBm}, i.e. with the Molchan-Golosov kernel \eqref{e:Molchan_kernel}. Then, the formulas for $U_0$ and $R_0$ are given in Proposition \ref{p:expVolterra} with particular kernel \eqref{e:Molchan_kernel}, autocovariance function \eqref{e:fBm_acov} and $\hat{r}(t,s|u)$ as in Theorem \ref{t:prediction_law}. In this case, we do not give the formulas for $U_0$ and $R_0$ after substituting the Molchan-Golosov kernel, since these formulas are too long. However, it is worth to mention that the formulas are explicit and numerical evaluation requires only the computation of some multiple Gaussian integrals.  

\end{example}


\section{Numerical comparison of approximation formula}\label{sec:numerics}

In this section, we focus on numerical aspects of the introduced approximation formula. We detail on its numerical implementation and a comparison with the Monte Carlo (MC) simulation framework introduced by \cite{Bennedsen17} will be provided.

In the second part of this section, we also introduce two interesting outcome analysis for rBergomi model. In particular, we show how the model can be efficiently calibrated using the approximation formula to short maturity smiles. We remark that classical SV models (e.g. Heston model) might fail to fit the short term smiles, unless they exploit high volatility of volatility levels for which they would be typically inconsistent with the long term skew of the volatility surface. 

In what follows, we will inspect the approximation quality for rBergomi model and time to maturity / volatility of volatility $\xi$ scenarios. Based on the nature of error terms (see Appendix \ref{sec:error_terms}) those two factors should play prominent role when it comes to approximation quality.

\subsection{On implementation of the approximation formula}
We note that for the models studied in this paper, we have obtained either a semi-closed form or analytical formula for standard Wiener case (H=0.5). Moreover, for the class of exponential fractional models -- represented by the $\alpha$RFSV model -- we only need to numerically evaluate multiple integrals in $R_0$ and $U_0$.

In our case, this was done using a trapezoid quadrature routine -- not necessarily the most efficient approach, but easy to implement. We used a discretisation\footnote{Typically we used from 1000 up to 27000 points for 3D integrals.} of integrands such that the numerical error doesn't affect the results in a significant way. I.e. to be lower than standard MC errors when comparing to simulated prices or lower than the expected approximation error.

For benchmarking we use a first-order hybrid MC scheme introduced by \cite{Bennedsen17} alongside 50000 MC sample paths. Similarly to the implementation of the approximation formula, we remark that this scheme could be also improved as described in \cite{McCrickerd18}.

\subsection{Sensitivity analysis for rBergomi ($\alpha=1$) approximation w.r.t. increasing $\xi$ and time to maturity $\tau$}
In this section, we illustrate the approximation quality for European call options under various model regimes / data set properties as described in Table \ref{Tbl:num_sen}. We use option expiries up to 1Y -- we are expecting a loss of approximation quality, based on the nature of the approximation formula. Since we utilized a first order approximation arguments with respect to volatility of volatility, we are also expecting more pronounced differences between MC simulations and the formula for large values of $\xi$.
\begin{table}[h]
\centering
\begin{tabular}{cc}
\toprule
Model params & Values \\
\midrule
$\xi$ & $\lbrace 10\% , 50\% , 100\% \rbrace$\\
$\sigma_0$ & 8\% \\
$\rho$ &  -20\% \\
$H$ & 0.1 \\
\toprule
Data set specifics & Values \\
\midrule
Moneyness & 70\% -- 130\% with 5\% step\\
Time to maturity & \{1M, 3M, 6M, 1Y\}\\
\bottomrule
\end{tabular}
\caption{Model / data settings for sensitivity analyses.}\label{Tbl:num_sen}
\end{table}

\begin{figure}[ht]
\begin{center}
\includegraphics[width=1.05\textwidth]{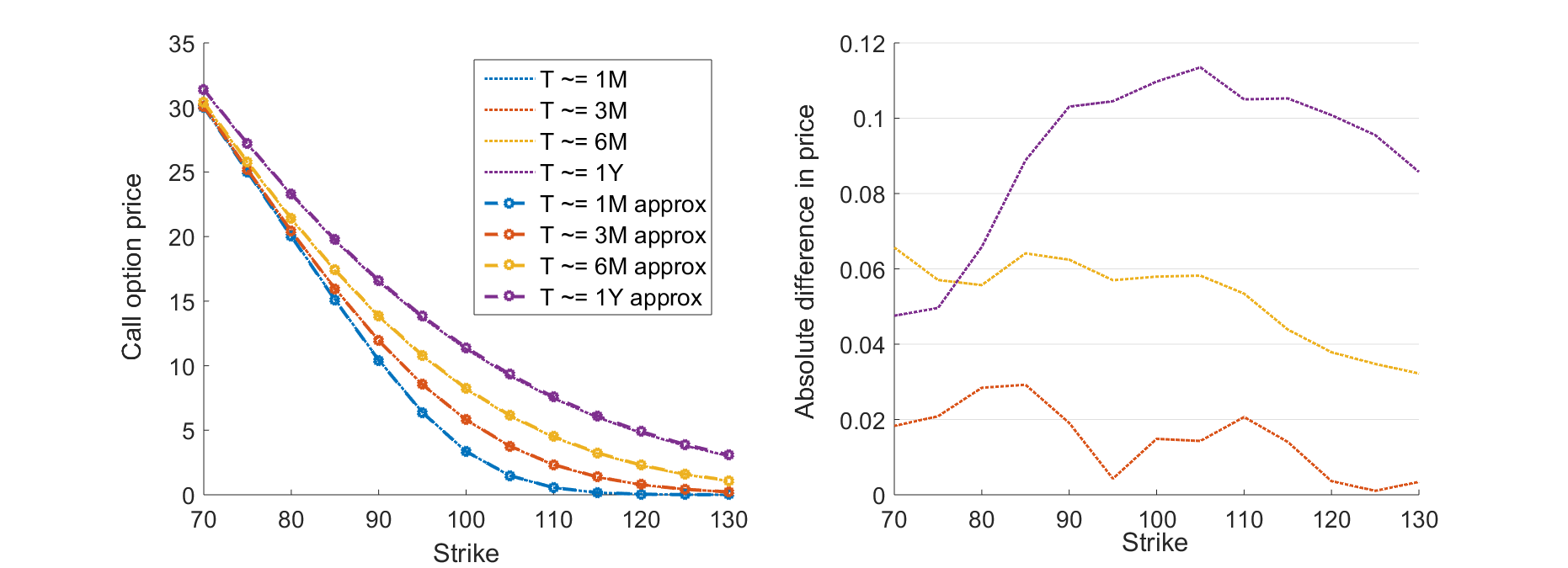}
\caption{Comparison of call option fair values calculated by MC simulations and by the approximation formula for the rBergomi model. (Example~\ref{ex:afbm} with $\alpha=1$ and $\varepsilon = 0$. Data and parameter values are: $v_0 = 8\%,~\xi=10\%, \rho = -20\%, H=0.1$) }\label{fig:rBergomi}
\end{center}
\end{figure}

In Figure \ref{fig:rBergomi}, we illustrate the approximation quality of the rBergomi approximation for low $\xi$ values. We can observe an expected behaviour: a very good match upto 3M expiry and almost linear deterioration of the quality with increasing $\tau$. Also the approximation formula provides a similar scale of errors across the tested moneyness.

For different moneyness regimes and 1M time to maturity, we obtained the following discrepancies between the MC trials and the introduced formula, measured in the relative option fair value (FV)\footnote{Relative FV is the absolute option fair value divided by the initial spot price.}:

\begin{center}

\begin{tabular}{c|ccc}
\toprule
&  \multicolumn{3}{c}{Differences in relative FVs} \\
\midrule
Spot moneyness & $\xi =10\%$ &$\xi = 50\%$ & $\xi = 100\%$ \\
\midrule
80\% & 4.5e-04 & -8.1e-05 & -0.29928 \\
90\% &3.9e-04 & 2.6e-05 & -0.02797 \\
100\% &2.3e-04 & 7.2e-04 & 0.95436 \\
110\% &-1.5e-05 & -7.7e-05 & 0.09690 \\ 
120\% &-1.2e-05 & -2.7e-04 & -0.46417 \\
\bottomrule

\end{tabular} 

\end{center}

In the table above, we can see reasonable approximation error measures which fell below standard 1 MC error for $\xi =10\%$ and $\xi =50\%$ regimes. Due to the theoretical properties of the approximation formula, we observe significant deterioration for high volatility of volatility regimes. This also depends on the time to maturity of the approximated option -- the shorter maturity we have, the higher $\xi$ we can allow to obtain reasonable approximation errors (i.e. of the order 1e-04 and lower in terms of FV).

Although the introduced approximation is typically not suitable for calibrations to the whole volatility surface -- due to the deterioration of approximation quality when increasing time to maturity -- we will illustrate how it can significantly speed-up MC calibration to the provided forward at-the-money (ATMF) backbone.

\subsection{Short-tenor calibration and a hybrid calibration to ATMF backbone}

Unlike previous analyses, which were based on artificial data / model parameter values, we inspect an application fo the formula on the calibration to real option market data. In particular, we utilize four data sets of AAPL options which were analysed in detail by \cite{PospisilSobotkaZiegler18ee}. Descriptive statistics of the data sets are provided in Table \ref{Tbl:data}. The following calibration test trials will be considered. 

\begin{itemize}
\item Calibration to short maturity smiles: \\

This should illustrate how well the model can fit short maturity smiles using the introduced approximation formula without exploiting too high volatility of volatility values ($\xi$). For each data set we selected the shortest maturity slice with more than one traded option. The values were not interpolated by any model, i.e. we calibrated to discrete close mid-prices of traded options. We also confirm, that both MC simulation and the formula reprice the smile with the final calibrated parameters without significant differences.

\item Hybrid calibration to the ATMF backbone: \\

In the second trial, we calibrate to each data set ATMF backbone. We note that because we have only a discrete set of traded options we might not have for each maturity an option with strike equal to the corresponding forward. Hence, we take an option with the closest strike to the forward value for each expiry. We use the proposed approximation formula only for $\tau < 0.2$, for longer time to maturities we price by MC simulations.

\end{itemize}

In both cases, the calibration routine was formulated as a standard least-square optimization problem. I.e. to obtain calibrated parameters, we numerically evaluated
\begin{equation}
\hat{\Theta} = \arg \min f(\Theta) = \arg \min \sum\limits_{i=1}^{N} \left[ \text{Mid}_i -\text{rBergomi}_i(\Theta) \right]^2,
\end{equation}
where $N$ is the total number of contracts for the calibration, $\text{Mid}_i$ is the mid price of the $i^{th}$ option and $\text{rBergomi}_i(\Theta)$ represents the corresponding model price based on parameter set $\Theta$. The model price is either obtained by the approximation formula or by means of MC simulations otherwise. The optimization is performed using Matlab's local search trust region optimizer which also needs an initial guess to start with.
\begin{table}[h]
\centering
{\small
\begin{tabular}{lcccc}
\toprule
 & Data  \# 1 & Data \# 2 & Data  \# 3& Data \# 4\\
 \midrule
 Date (all EOD)& 1-Apr-2015 & 15-Apr-2015 & 1-May-2015 & 15-May-2015 \\
 Moneyness range & 34\%--157\%  & 45\%--154\%  & 31\%--151\% & 33\%--151\% \\
 Time to maturity [Yrs] & 0.12-1.81  & 0.08--1.77 & 0.04--1.73 & 0.02--1.69  \\
 Total nb. of contracts & 113  & 158  & 201  & 194  \\

\bottomrule
\end{tabular}
\caption{Data on AAPL options used in calibration trials}\label{Tbl:data}
}
\end{table}

All following results will be quoted in relative FV: e.g. $\text{rBergomi}_i(\Theta)/S_0$ and also differences between market and the calibrated model will be denoted using this measure. For the calibration to the whole surface of European options, errors in FV below $0.5\%$ are typically considered to be acceptable, whereas anything exceeding $1\%$ difference is considered as a significant model inconsistency.

Firstly, we display the results for the short-maturity calibration. In Figure \ref{fig:rBergomi_shorT_1}, we illustrate that even with a not well suited initial guess for Data \# 4, we can obtain satisfactory results (i.e. errors significantly below 0.5\% mark). In fact, for all tested data sets, obtained values of the calibrated parameters were not very sensitive to the initial guess (only a number of iterations differed). This is a desired feature which typically is not present under classical SV models, see e.g. \cite{MrazekPospisilSobotka16ejor}. For two other sets we have obtained qualitatively similar fitting errors, however for the data set \# 3 we have retrieved 4 errors (out of 22) with absolute FV difference greater than 0.5\% and two even greater than 1\%, see Figure \ref{fig:rBergomi_shorT_2}. We conclude that this was partly caused due to high $\xi$ values compared to other three calibration trials and also slightly longer maturity  -- the data set \#3 includes only one option at the shortest maturity, hence the second shortest was used. Still we can conclude that the obtained errors (also verified by using MC simulations) are overall acceptable, although might not be optimal.

\begin{figure}[ht]
    \centering
    \begin{subfigure}[b]{1\textwidth}
        \includegraphics[width=\textwidth]{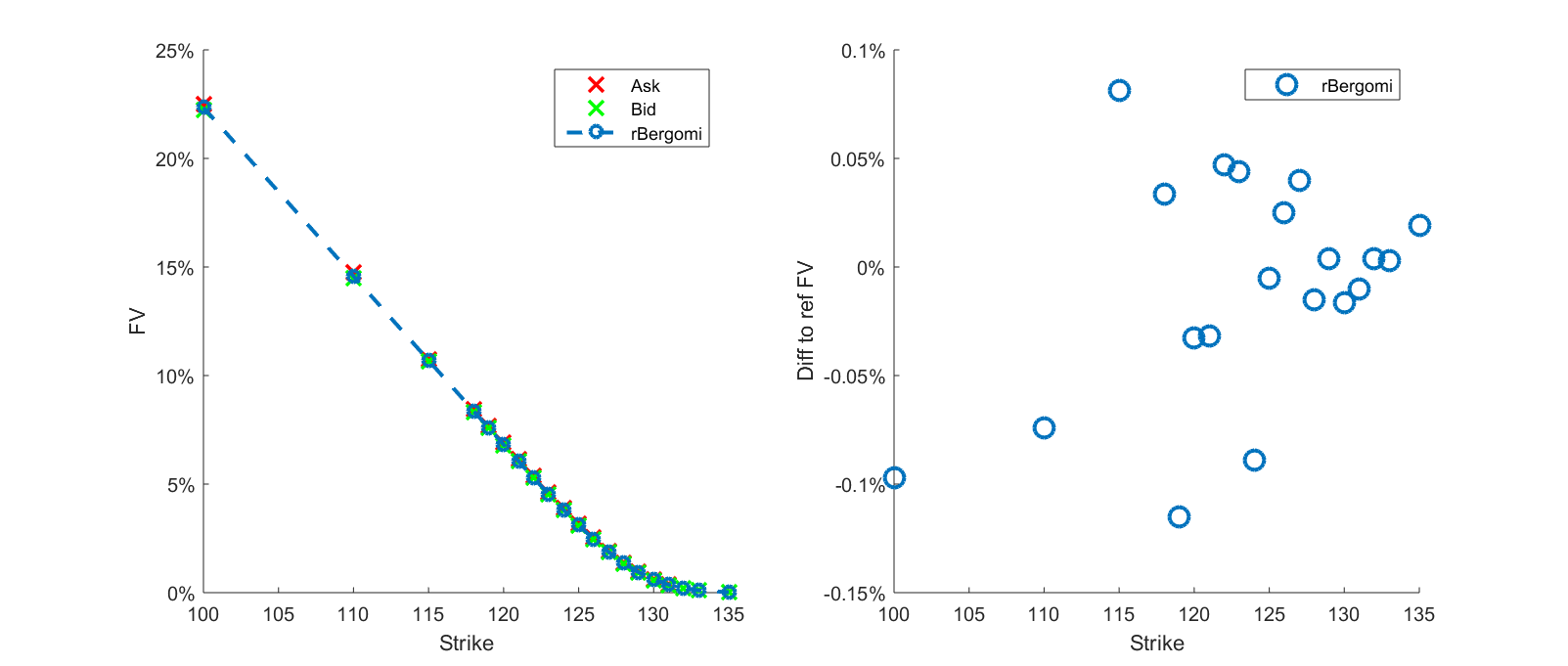}
        \caption{Comparison of calibrated rBergomi and market data (Data \# 4)}
        \label{fig:Calib_May15}
    \end{subfigure}
    ~ 
    \begin{subfigure}[b]{0.51\textwidth}
        \includegraphics[width=\textwidth]{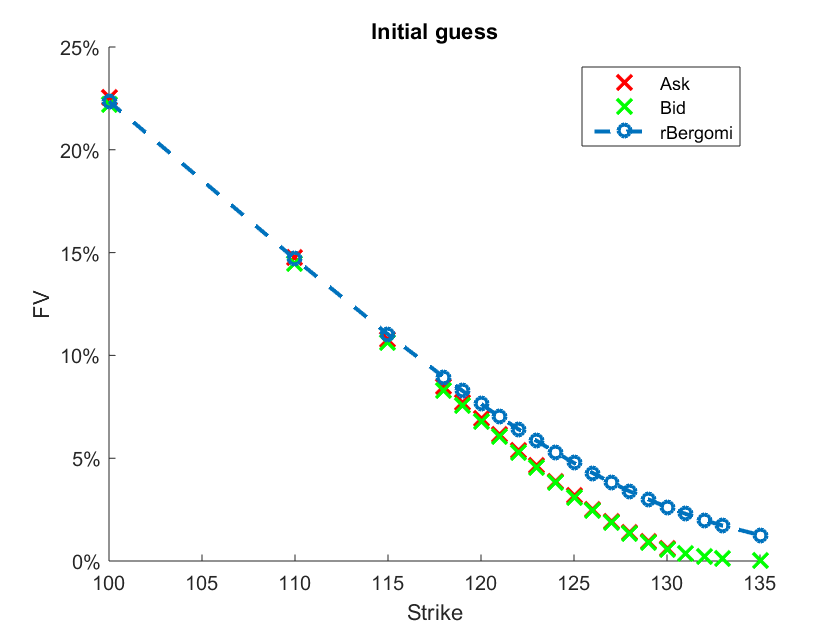}
        \caption{Errors with respect to the initial guess}
        \label{fig:init_May15}
    \end{subfigure}
    \begin{subfigure}[b]{0.38\textwidth}     

        \begin{minipage}{\textwidth}
        \vspace*{-5cm} 
        \textbf{Calibrated params:} \\
        \vspace*{0.1cm} \hrule
        \centering \vspace*{0.1cm}
        $\sigma_0 = 3.41 \%$ \\
        $\xi =  39.45 \%$ \\
        $\rho = -98.80 \%$ \\
        $H = 0.3153$ \\
        \vspace*{0.1cm} \hrule
        \vspace*{0.1cm}
        \textbf{MSE = 0.0883}
\end{minipage}         

    \end{subfigure}
    \caption{Calibration results for rBergomi - short maturity smiles (Data \# 4)}\label{fig:rBergomi_shorT_1}
\end{figure}

We have observed that at least short maturity smiles (< 1M) can be efficiently calibrated using the proposed approximation formula, while for expiries greater than 1M, one would need to stay within a low volatility of volatility regime, otherwise the discrepancies would lead to a non-optimal solution when recomputed using a more precise (but much more costly) MC simulations. To illustrate efficiency of the approximation will show how the approximation can speed-up ATMF-backbone calibration.

\begin{figure}[ht]
    \centering
    \begin{subfigure}[b]{1\textwidth}
        \includegraphics[width=\textwidth]{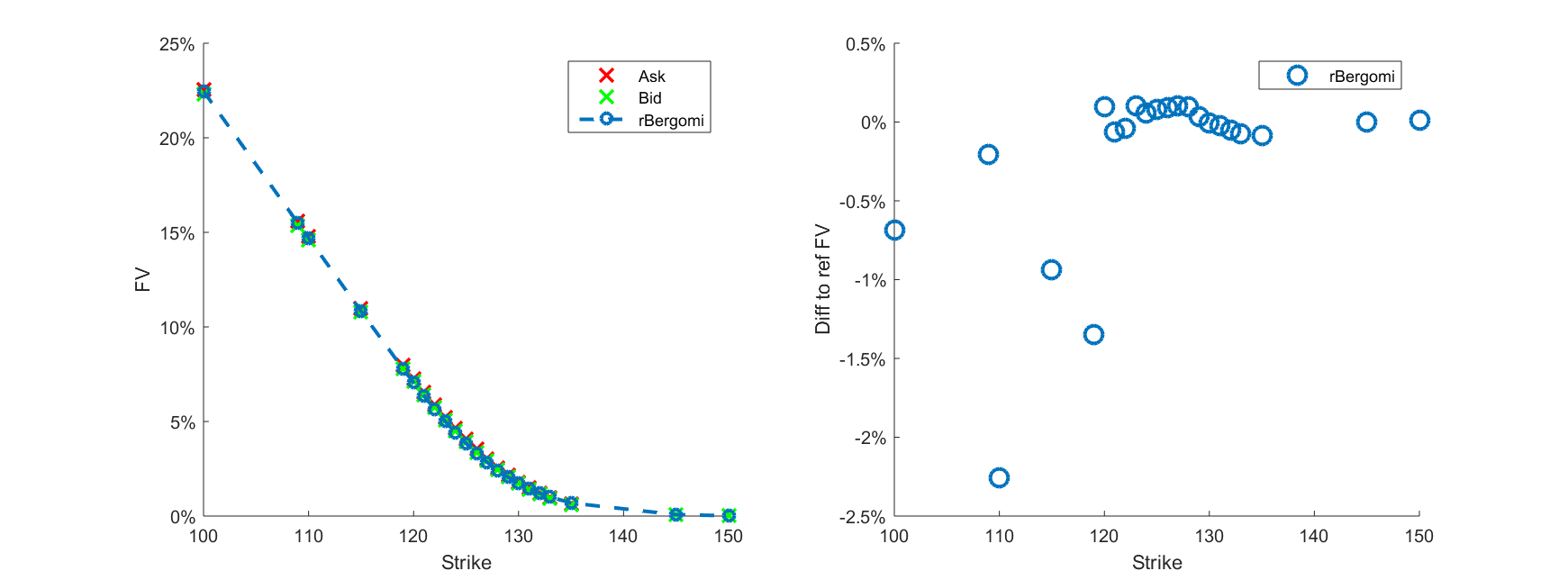}
        \caption{Comparison of calibrated rBergomi and market data (Data \# 3)}
        \label{fig:Calib_May1}
    \end{subfigure}
    ~ 
    \begin{subfigure}[b]{0.51\textwidth}
        \includegraphics[width=\textwidth]{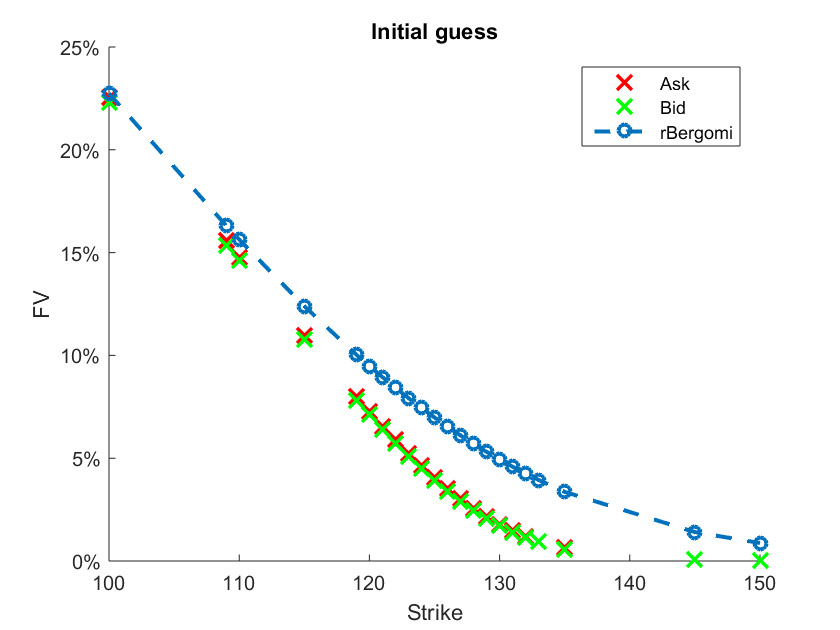}
        \caption{Errors with respect to the initial guess}
        \label{fig:init_May1}
    \end{subfigure}
    \begin{subfigure}[b]{0.38\textwidth}     

        \begin{minipage}{\textwidth}
        \vspace*{-5cm} 
        \textbf{Calibrated params:} \\
        \vspace*{0.1cm} \hrule
        \centering \vspace*{0.1cm}
        $\sigma_0 = 4.57 \%$ \\
        $\xi =  81.36 \%$ \\
        $\rho = -98.99 \%$ \\
        $H = 0.3936$ \\
        \vspace*{0.1cm} \hrule
        \vspace*{0.1cm}
        \textbf{MSE = 13.9538}
\end{minipage}         

    \end{subfigure}
    \caption{Calibration results for rBergomi - short maturity smiles (Data \# 3)}\label{fig:rBergomi_shorT_2}
\end{figure}

We will now inspect the hybrid calibration where we switch between the approximation and a MC pricer based on properties of options being priced. In particular, we focus on data set from 15$^{th}$ May (Data \# 4) and for $\tau < 0.2$ we will use the approximation formula and MC simulations otherwise. For the calibrated parameters, we will also measure the time spent computing FVs by each pricer. For completeness, we remark that both implementations are not perfect. I.e. for MC simulations under the rBergomi model, the "turbocharging" improvements were introduced by \cite{McCrickerd18} recently. On the other hand, numerical integrations in the approximation could be performed by some adaptive quadrature and could be vectorized to improve the computation speed. 

In Figure \ref{fig:rBergomiATMF}, we illustrate calibration fit to the ATMF backbone of the option price surface. We conclude that we have retrieved similar errors for both the prices computed using the proposed approximation and the longer maturity option prices quantified by MC simulations. The final fit of the calibrated model (recomputed by MC simulations) is very good, especially considering that the studied model has only 4 parameters. Moreover, only a fraction of the time spent by MC pricer was needed to computed all FV using the approximation formula. In particular, $98.43\%$ of the pricing time\footnote{Excluding any data loading / manipulation routines.} we were computing MC simulation estimates of FVs. We also note that $7/9$ of total evaluations were computed by the approximation formula.

\begin{figure}[ht]
\begin{center}
\includegraphics[width=1.1\textwidth]{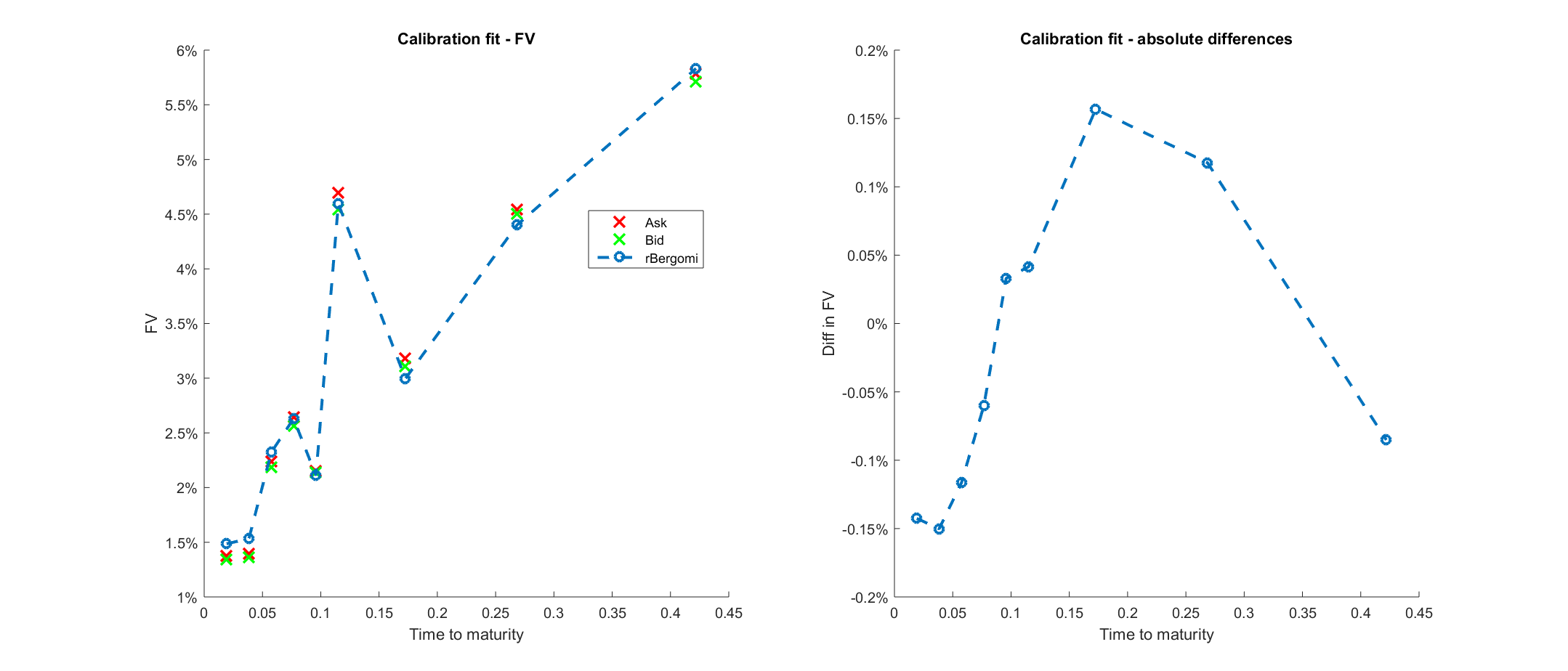}
\caption{ATMF calibration results when combining approximation formula ($\tau <0.2$) and MC simulations.}\label{fig:rBergomiATMF}
\end{center}
\end{figure}
\newpage
\section{Conclusion}\label{sec:conclusion}

In previous sections we studied an approximation approach for the pricing of European options under rough stochastic volatility dynamics. Our approach is based on the option price decomposition results obtained by \cite{Alos12} for the standard Heston SV model and on its recent generalisation to other SV models by \cite{MerinoPospisilSobotkaVives18ijtaf}. 
The main contribution of our research is to derive pricing formulas suitable for various practical applications by applying the general decomposition on a class of Volterra volatility models. Our main focus is laid on the rough volatility models introduced by \cite{Gatheral18} and \cite{Bayer16}. 

In particular, a prediction law for Gaussian Volterra processes was proved and an adapted projection of future volatility was obtained in Section \ref{sec:model} for the class of exponential Volterra volatility models, where the volatility process can be expressed as a positive $L^1$ function of the time and the Volterra process $Y_t$, i.e. $\sigma_t = g(t,Y_t)$. We focused on one particular example of Gaussian Volterra processes, namely on a (rough) fractional Brownian motion $B^H_t$.
"Roughness" of sample paths is determined by the Hurst parameter, $H \in (0,1/2]$, where for $H=0.5$ we would recover a standard Wiener process.

The pricing formula for European options, which is numerically tractable, is then derived under a newly introduced $\alpha$RFSV model, i.e. for 
$$\sigma_t =\sigma_{0} \exp\left\{\xi B^H_t \NEW{- \frac12 \alpha\xi^2 r(t)}\right\},$$ 
where $\alpha \in [0,1]$ and $r(t)$ is the corresponding autocovariance function defined in Section \ref{sec:model}. This newly introduced model seems to have interesting special cases: 
\begin{itemize}
\item for $\alpha = 0$ it reverts to a simple exponential RFSV model
\item and for $\alpha=1$ to rBergomi model respectively.
\end{itemize}

Both special cases of the $\alpha$RFSV model are studied and praised for their surprising consistency with various financial markets in \citep{Bayer16} who use the the Bergomi--Guyon expansion to study the ATM skew approximation of the impied volatilities. However, the authors admit that \emph{``the Bergomi--Guyon expansion does not converge with values of $\eta$ consistent with the SPX volatility surface, so the Bergomi--Guyon expansion is not useful in practice for calibration of the rBergomi model.''} This motived us to derive an approximation formula that will be useful in practice for calibration of the studied models to real market data.
 
In Section \ref{sec:numerics}, the newly obtained approximation formula was compared to the MC pricing approach introduced by \cite{Bennedsen17} and \cite{McCrickerd18}. This enabled us to numerically verify the obtained solution, to quantify its approximation errors under various settings and, last but not least, to comment on suitability of the rBergomi model for calibration tasks to real market data based on AAPL stock options\footnote{These data sets were analysed and described in \cite{PospisilSobotkaZiegler18ee}}. The following conclusions were drawn:

\begin{itemize}
\item The approximation error is well behaved for short maturities (typically for less than 1M) and the error increases with time to maturity and $\xi$ parameter.
\item For medium-term expiries we are able to obtain well approximated prices only under low $\xi$ regimes.
\item Although, the approximation under a rough Volterra process involves several numerical integration procedures, it is much faster than MC simulation approach implemented in the same environment\footnote{The numerical trials were implemented in MATLAB environment, see \url{https://www.mathworks.com}.}. For the standard Wiener case (and in particular for the original Bergomi model), we can have an analytical pricing formula, but also more efficient and well developed MC simulation schemes.
\item Considering that the modelling approach studied has only few parameters, we were able to fit the sample market data surprisingly well. For calibration to short maturity smiles, we can use just the approximation formula -- this was verified by recomputing calibration errors using MC simulations. For calibration to the whole surface, one can utilize a newly introduced hybrid scheme which consist of a combination of approximation and simulation techniques. The idea is quite simple -- to use approximation formula for low maturities or for low $\xi$ values and MC simulations for the remaining computations. Suitability of this scheme was judged by a simple calibration to an ATMF-like backbone. We retrieved a well calibrated model, while saving a significant computational time compared to the calibration based on MC simulations only.
\end{itemize}

However, we remark that implementation of the approximation could be further improved -- for simplicity we used a simple trapezoidal quadrature to numerically evaluate integrals appearing in $U_t$ and $R_t$ expressions. 

Another possible improvement could focus on postulating variance dynamics in terms of a forward variance curve -- to ensure consistency with variance swaps. This is out-of-scope for the current submission and left for further research. However, some of the main ingredients for forward variance curves under the $\alpha$RFSV model are already introduced in this manuscript.

Also to fit some of the pronounced forward variance curves, one might need a more complex term structure -- a simple exponential drift term under the rBergomi model might not be sufficiently flexible. In this case, one can utilize our approach to obtain the approximation formula for a new rough volatility model. In fact, once the new volatility function $g$ is postulated it should be only a matter of algebraic operations to obtain the corresponding approximation formula. 

\section*{Funding}

The work of Jan Posp\'{\i}\v{s}il was partially supported by the Czech Science Foundation (GA\v{C}R) grant no. GA18-16680S ``Rough models of fractional stochastic volatility''.
The work of Josep Vives was partially supported by Spanish grant MEC MTM 2016-76420-P. 

\section*{Acknowledgements}

Computational resources were provided by the CESNET LM2015042 and the CERIT Scientific Cloud LM2015085, provided under the programme ``Projects of Large Research, Development, and Innovations Infrastructures''.

\appendix
\section{Decomposition formulas in the general model}\label{sec:error_terms}

In this appendix, we give the error terms for a decomposition of the general model.

The term (I) can be decomposed as
\begin{eqnarray*}
&&\frac{\rho}{2}\mathbb{E}\left[\int^{T}_{t}e^{-r(u-t)}\Lambda \Gamma BS(u,X_{u},v_{u})\sigma_{u}\d\la W,M\ra_{u}\right]-\Lambda \Gamma BS(t,X_{t},v_{t})U_{t}\\
&=&\frac{1}{8}\mathbb{E}\left[\int^{T}_{t}e^{-r(u-t)}\Lambda \Gamma^{3} BS(u,X_{u},v_{u})U_{u}\d\la M,M\ra_{u}\right]\\ 
&+&\frac{\rho}{2}\mathbb{E}\left[\int^{T}_{t}e^{-r(u-t)}\Lambda^{2} \Gamma^{2} BS(u,X_{u},v_{u})U_{u}\sigma_{u}\d\la W,M\ra_{u}\right]\\ 
&+&\rho\mathbb{E}\left[\int^{T}_{t}e^{-r(u-t)}\Lambda^{2} \Gamma BS(u,X_{u},v_{u})\sigma_{u} \d\la W,U\ra_{u}\right]\\ 
&+&\frac{1}{2}\mathbb{E}\left[\int^{T}_{t}e^{-r(u-t)}\Lambda \Gamma^{2} BS(u,X_{u},v_{u})\d\la M,U\ra_{u}\right].
\end{eqnarray*}

The term (II) can be decomposed as
\begin{eqnarray*}
&&\frac{1}{8}\mathbb{E}\left[\int^{T}_{t}e^{-r(u-t)}\Gamma^{2} BS(u,X_{u},v_{u})\d\la M,M\ra_{u}\right]-\Gamma^{2} BS(t,X_{t},v_{t})R_{t}\\ 
&=&\frac{1}{8}\mathbb{E}\left[\int^{T}_{t}e^{-r(u-t)}\Gamma^{4} BS(u,X_{u},v_{u})R_{u}\d\la M,M\ra_{u}\right]\\ 
&+&\frac{\rho}{2}\mathbb{E}\left[\int^{T}_{t}e^{-r(u-t)}\Lambda \Gamma^{3} BS(u,X_{u},v_{u})R_{u}\sigma_{u}\d\la W,M\ra_{u}\right]\\ 
&+&\rho\mathbb{E}\left[\int^{T}_{t}e^{-r(u-t)}\Lambda\Gamma^{2} BS(u,X_{u},v_{u})\sigma_{u} \d\la W,R\ra_{u}\right]\\ 
&+&\frac{1}{2}\mathbb{E}\left[\int^{T}_{t}e^{-r(u-t)}\Gamma^{3} BS(u,X_{u},v_{u})\d\la M,R\ra_{u}\right].
\end{eqnarray*}

\section{Upper-bound for decomposition formula for the exponential Wiener volatility model}\label{sec:Upp_bound_Wiener}

In this appendix we obtain the upper-bound for the decomposition formula for the exponential Wiener volatility model in Example \ref{ex:ExpWiener}. 

\subsection{Upper-bound for term (I)}\label{Upp_bound_Wiener_(I)}
For matter of convenience, we define the function
\begin{align*}
\chi_1(t,T\NEW{,\alpha}) &:= \int^{T}_{t} \exp\left\{\frac{1}{2}(9-3\NEW{\alpha})\xi^{2} (s-t) \right\} \left[\exp\left\{(2\NEW{-\alpha})\xi^{2}  (T-s)\right\}-1\right] \d s.
\intertext{We can rewrite $U_{t}$ as}
U_{t} &= \frac{\rho\sigma^{3}_{t}}{(2\NEW{-\alpha})\xi}\chi_1(t,T\NEW{,\alpha}).
\intertext{It is easy to find that}
\d U_{t} 
&= \frac{\rho \d\sigma^{3}_{t}}{(2\NEW{-\alpha})\xi}\chi_1(t,T\NEW{,\alpha}) + \frac{\rho\sigma^{3}_{t}}{(2\NEW{-\alpha})\xi}\chi_1 ' (t,T\NEW{,\alpha}) \d t \\
&= \frac{\rho \left(3\xi\sigma^{3}_{t} \d W_{t} + \frac{1}{2}(18-3\NEW{\alpha})\xi^{2} \sigma^{3} \d t\right)}{(2\NEW{-\alpha})\xi}\chi_1(t,T\NEW{,\alpha}) + \frac{\rho\sigma^{3}_{t}}{(2\NEW{-\alpha})\xi}\chi_1 ' (t,T\NEW{,\alpha}) \d t .
\end{align*}
If $\alpha\geq 0$, we can find an upper-bound for $\chi_1(t,T\NEW{,\alpha})$ which is
\begin{eqnarray*}
\chi_1(t,T\NEW{,\alpha})&\leq& \int^{T}_{t} \exp\left\{\frac{9}{2}\xi^{2} (s-t) \right\} \left[\exp\left\{2\xi^{2}  (T-s)\right\}-1\right] \d s\\
&=&\frac{2}{45 \xi ^2} \left(-9 \exp\left\{2 \xi ^2 (T-t) \right\}+4 \exp\left\{\frac{9}{2} \xi ^2 (T-t) \right\}+5\right).
\end{eqnarray*}
We can re-write the decomposition formula as
\begin{eqnarray*}
&&\frac{\rho}{2}\mathbb{E}_{t}\left[\int^{T}_{t}e^{-r(u-t)}\Lambda \Gamma BS(u,X_{u},v_{u})\sigma_{u}\d\left\langle W,M\right\rangle_{u}\right]-\Lambda \Gamma BS(t,X_{t},v_{t})U_{t}\\ 
&=&\frac{1}{8}\mathbb{E}_{t}\left[\int^{T}_{t}e^{-r(u-t)}\left(\partial^{5}_{x}-2\partial^{4}_{x}+\partial^{3}_{x}\right) \Gamma BS(u,X_{u},v_{u})U_{u}\d\left\langle M,M\right\rangle_{u}\right]\\ 
&+&\frac{\rho}{2}\mathbb{E}_{t}\left[\int^{T}_{t}e^{-r(u-t)}\left(\partial^{4}_{x}-\partial^{3}_{x}\right) \Gamma BS(u,X_{u},v_{u})U_{u}\sigma_{u}\d\left\langle W,M\right\rangle_{u}\right]\\ 
&+&\rho\mathbb{E}_{t}\left[\int^{T}_{t}e^{-r(u-t)}\partial^{2}_{x}\Gamma BS(u,X_{u},v_{u})\sigma_{u} \d\left\langle W,U\right\rangle_{u}\right]\\ 
&+&\frac{1}{2}\mathbb{E}_{t}\left[\int^{T}_{t}e^{-r(u-t)}\left(\partial^{3}_{x}-\partial^{2}_{x}\right)\Gamma BS(u,X_{u},v_{u})\d\left\langle M,U\right\rangle_{u}\right].
\end{eqnarray*}
Applying Lemma \ref{lemaclau} 
and using the definition of $a_{u}$, we obtain
\begin{eqnarray*}
&&\left|\frac{\rho}{2}\mathbb{E}_{t}\left[\int^{T}_{t}e^{-r(u-t)}\Lambda \Gamma BS(u,X_{u},v_{u})\sigma_{u}\d\left\langle W,M\right\rangle_{u}\right]-\Lambda \Gamma BS(t,X_{t},v_{t})U_{t}\right|\\ 
&\leq&\frac{C}{8}\mathbb{E}_{t}\left[\int^{T}_{t}e^{-r(u-t)}\left(\frac{1}{a^6_{u}}+\frac{2}{a^5_{u}}+\frac{1}{a^4_{u}}\right)U_{u}\d\left\langle M,M\right\rangle_{u}\right]\\ 
&+&\frac{C\rho}{2}\mathbb{E}_{t}\left[\int^{T}_{t}e^{-r(u-t)}\left(\frac{1}{a^5_{u}}+\frac{1}{a^4_{u}}\right)U_{u}\sigma_{u}\d\left\langle W,M\right\rangle_{u}\right]\\ 
&+&C\rho\mathbb{E}_{t}\left[\int^{T}_{t}e^{-r(u-t)}\frac{1}{a^3_{u}}\sigma_{u} \d\left\langle W,U\right\rangle_{u}\right]\\ 
&+&\frac{C}{2}\mathbb{E}_{t}\left[\int^{T}_{t}e^{-r(u-t)}\left(\frac{1}{a^4_{u}}+\frac{1}{a^3_{u}}\right)\d\left\langle M,U\right\rangle_{u}\right].
\end{eqnarray*}
Noting that $a_{u}=\sigma_{u}\phi^{1/2}(u,T,\alpha)$, where $\phi$ was defined in \eqref{e:phi},
\begin{eqnarray*}
&&\left|\frac{\rho}{2}\mathbb{E}_{t}\left[\int^{T}_{t}e^{-r(u-t)}\Lambda \Gamma BS(u,X_{u},v_{u})\sigma_{u}\d\left\langle W,M\right\rangle_{u}\right]-\Lambda \Gamma BS(t,X_{t},v_{t})U_{t}\right|\\ 
&\leq&\frac{C\rho\xi}{2(2-\alpha)}\mathbb{E}_{t}\left[\int^{T}_{t}e^{-r(u-t)}\left(\frac{\sigma_{u}}{\phi(u,T,\alpha)}+\frac{2\sigma^{2}_{u}}{\phi^{1/2}(u,T,\alpha)}+\sigma^{3}_{u}\right)\chi_1(u,T,\alpha) \d u\right]\\ 
&+&\frac{C\rho^{2}}{(2-\alpha)}\mathbb{E}_{t}\left[\int^{T}_{t}e^{-r(u-t)}\left(\frac{\sigma_{u}}{\phi^{3/2}(u,T,\alpha)}+\frac{\sigma^{2}_{u}}{\phi(u,T,\alpha)}\right)\chi_1(u,T,\alpha) \d u\right]\\
&+&\frac{C\rho^{2} 3 }{(2-\alpha)}\mathbb{E}_{t}\left[\int^{T}_{t}e^{-r(u-t)}\frac{\sigma_{u}}{\phi^{3/2}(u,T,\alpha)}\chi_1(u,T,\alpha) \d u\right]\\ 
&+&\frac{3C\rho\xi}{(2-\alpha)}\mathbb{E}_{t}\left[\int^{T}_{t}e^{-r(u-t)}\left(\frac{\sigma_{u}}{\phi(u,T,\alpha)}+\frac{\sigma^{2}_{u}}{\phi^{1/2}(u,T,\alpha)}\right)\chi_1(u,T,\alpha) \d u\right].
\end{eqnarray*}
Being $\sigma_{u}$ the only stochastic component, we can get the expectation inside. Each power of $\sigma_{u}$ has a different forward value, in this case, we can bound all terms by $\exp\left\{\frac{9}{2}\xi^{2} (u-t) \right\}$. We have that
\begin{eqnarray*}
&&\left|\frac{\rho}{2}\mathbb{E}_{t}\left[\int^{T}_{t}e^{-r(u-t)}\Lambda \Gamma BS(u,X_{u},v_{u})\sigma_{u}\d\left\langle W,M\right\rangle_{u}\right]-\Lambda \Gamma BS(t,X_{t},v_{t})U_{t}\right|\\ 
&\leq&\frac{C\rho\xi}{2(2-\alpha)}\int^{T}_{t}e^{-r(u-t)}\exp\left\{\frac{9}{2}\xi^{2}(u-t)\right\}\left(\frac{\sigma_{t}}{\phi(u,T,\alpha)}+\frac{2\sigma^{2}_{t}}{\phi^{1/2}(u,T,\alpha)}+\sigma^{3}_{t}\right)\chi_1(u,T,\alpha) \d u\\ 
&+&\frac{C\rho^{2}}{(2-\alpha)}\int^{T}_{t}e^{-r(u-t)}\exp\left\{\frac{9}{2}\xi^{2} (u-t) \right\}\left(\frac{\sigma_{t}}{\phi^{3/2}(u,T,\alpha)}+\frac{\sigma^{2}_{t}}{\phi(u,T,\alpha)}\right)\chi_1(u,T,\alpha) \d u\\ 
&+&\frac{C\rho^{2} 3 }{(2-\alpha)}\int^{T}_{t}e^{-r(u-t)}\exp\left\{\frac{9}{2}\xi^{2} (u-t) \right\}\frac{\sigma_{t}}{\phi^{3/2}(u,T,\alpha)}\chi_1(u,T,\alpha) \d u\\ 
&+&\frac{3C\rho\xi}{(2-\alpha)}\int^{T}_{t}e^{-r(u-t)}\exp\left\{\frac{9}{2}\xi^{2} (u-t) \right\}\left(\frac{\sigma_{t}}{\phi(u,T,\alpha)}+\frac{\sigma^{2}_{t}}{\phi^{1/2}(u,T,\alpha)}\right)\chi_1(u,T,\alpha) \d u.
\end{eqnarray*}
Substituting $\phi(u,T,\alpha)$ and using the upper-bound for $\chi_1(u,T,\alpha)$ when $\alpha\geq0$, we have
\begin{eqnarray*}
&&\left|\frac{\rho}{2}\mathbb{E}_{t}\left[\int^{T}_{t}e^{-r(u-t)}\Lambda \Gamma BS(u,X_{u},v_{u})\sigma_{u}\d\left\langle W,M\right\rangle_{u}\right]-\Lambda \Gamma BS(t,X_{t},v_{t})U_{t}\right|\\ \nonumber
&\leq&\frac{C\rho\xi}{45(2-\alpha)\xi^2}\int^{T}_{t}e^{-r(u-t)}\exp\left\{\frac{9}{2}\xi^{2}(u-t)\right\}\\
&\qquad&\Bigl(\sigma_{t}\frac{2\xi^{2}}{\left[\exp\left\{2\xi^{2}  (T-u)\right\}-1\right]}+2\sigma^{2}_{t}\frac{\sqrt{2}\xi}{\left[\exp\left\{(2-\alpha)\xi^{2}  (T-u)\right\}-1\right]^\frac{1}{2}}+\sigma^{3}_{t}\Bigr)\\
&\qquad& \left(-9 \exp\left\{2 \xi ^2 (T-u)\right\}+4 \exp\left\{\frac{9}{2} \xi ^2 (T-u)\right\}+5\right) \d u\\
&+&\frac{2C\rho^{2}}{45(2-\alpha) \xi^2}\int^{T}_{t}e^{-r(u-t)}\exp\left\{\frac{9}{2}\xi^{2} (u-t) \right\}\\
&\qquad&\Bigl(\sigma_{t}\frac{2^\frac{3}{2}\xi^{3}}{\left[\exp\left\{(2-\alpha)\xi^{2}  (T-u)\right\}-1\right]^\frac{3}{2}}+\sigma^{2}_{t}\frac{2\xi^{2}}{\left[\exp\left\{(2-\alpha)\xi^{2}  (T-u)\right\}-1\right]}\Bigr)\\
&\qquad&\left(-9 \exp\left\{2 \xi ^2 (T-u)\right\}+4 \exp\left\{\frac{9}{2} \xi ^2 (T-u)\right\}+5\right) \d u\\
&+&\frac{6C\rho^{2} }{45(2-\alpha)\xi^2}\int^{T}_{t}e^{-r(u-t)}\exp\left\{\frac{9}{2}\xi^{2} (u-t) \right\}\\
&\qquad&\sigma_{t}\frac{2^\frac{3}{2}\xi^{3}}{\left[\exp\left\{(2-\alpha)\xi^{2}  (T-u)\right\}-1\right]^\frac{3}{2}} \left(-9 \exp\left\{2 \xi ^2 (T-u)\right\}+4 \exp\left\{\frac{9}{2} \xi ^2 (T-u)\right\}+5\right) \d u\\
&+&\frac{6C\rho\xi}{45(2-\alpha)\xi^2}\int^{T}_{t}e^{-r(u-t)}\exp\left\{\frac{9}{2}\xi^{2} (u-t) \right\}\\
&\qquad&\left(\sigma_{t}\frac{2\xi^{2}}{\left[\exp\left\{(2-\alpha)\xi^{2}  (T-u)\right\}-1\right]}+\sigma^{2}_{t}\frac{\sqrt{2}\xi}{\left[\exp\left\{(2-\alpha)\xi^{2}  (T-u)\right\}-1\right]^\frac{1}{2}}\right)\\
&\qquad& \left(-9 \exp\left\{2 \xi ^2 (T-u)\right\}+4 \exp\left\{\frac{9}{2} \xi ^2 (T-u)\right\}+5\right) \d u.
\end{eqnarray*}
The above upper-bound error is difficult to interpret. In order to do this, we derive a Taylor expansion for one of the terms. Then, the following error behaviour is retrieved:
$$C\frac{ \rho \xi^3   \sigma_{t}  T^{3/2}}{6 (\alpha -2)^2} \left(\frac{32 \sqrt{2} \rho }{\sqrt{-(\alpha -2) \xi ^2}}+21 \sqrt{T}\right).$$

\subsection{Upper-bound for term (II)}\label{Upp_bound_Wiener_(II)}
For matter of convenience, we define the function
\begin{align*}
\chi_2(t,T\NEW{,\alpha}) &:= \int^{T}_{t}\exp\left\{(8\NEW{-2\alpha})\xi^{2} (s-t) \right\} \left[\exp\left\{(2\NEW{-\alpha})\xi^{2}  (T-s)\right\}-1\right]^{2} \d s
\intertext{We can re-write $R_{t}$ as}
R_{t} &= \frac{\sigma^{4}_{t}}{2(2\NEW{-\alpha})^{2}\xi^{2}}\chi_2(t,T\NEW{,\alpha}).
\intertext{It is easy to find that}
\d R_{t} 
&= \frac{\d\sigma^{4}_{t}}{2(2\NEW{-\alpha})^{2}\xi^{2}}\chi_2(t,T\NEW{,\alpha}) + \frac{\sigma^{4}_{t}}{2(2\NEW{-\alpha})^{2}\xi^{2}}\chi_2 ' (t,T\NEW{,\alpha}) \d t \\
&= \frac{4\xi\sigma^{4}_{t}\d W_{t} + 2(8\NEW{-\alpha})\sigma^{4}_{t} \d t}{2(2\NEW{-\alpha})^{2}\xi^{2}}\chi_2(t,T\NEW{,\alpha}) + \frac{\sigma^{4}_{t}}{2(2\NEW{-\alpha})^{2}\xi^{2}}\chi_2 ' (t,T\NEW{,\alpha}) \d t.
\end{align*}
If $\alpha\geq 0$, we can find an upper-bound for $\chi_2(t,T,\alpha)$ which is
\begin{eqnarray*}
\chi_2(t,T,\alpha)&\leq& \int^{T}_{t}\exp\left\{8\xi^{2} (s-t) \right\} \left[\exp\left\{2\xi^{2}  (T-s)\right\}-1\right]^{2} \d s\\
&=&\frac{1}{24\xi^{2}}\left(\exp\left\{2\xi^{2}  (T-t)\right\}-1\right)^{3}\left(\exp\left\{2\xi^{2}  (T-t)\right\}+3\right).
\end{eqnarray*}
We can re-write the decomposition formula as
\begin{eqnarray*}
&&\frac{1}{8}\mathbb{E}_{t}\left[\int^{T}_{t}e^{-r(u-t)}\Gamma^{2} BS(u,X_{u},v_{u})\d\left\langle M,M\right\rangle_{u}\right]-\Gamma^{2} BS(t,X_{t},v_{t})R_{t}\\
&=&\frac{1}{8}\mathbb{E}_{t}\left[\int^{T}_{t}e^{-r(u-t)}\left(\partial^{6}_{x}-3\partial^{5}_{x}+3\partial^{4}_{x}- \partial^{3}_{x}\right)\Gamma BS(u,X_{u},v_{u})R_{u}\d\left\langle M,M\right\rangle_{u}\right]\\ 
&+&\frac{\rho}{2}\mathbb{E}_{t}\left[\int^{T}_{t}e^{-r(u-t)}\left(\partial^{5}_{x}-2\partial^{4}_{x}+\partial^{3}_{x}\right) \Gamma BS(u,X_{u},v_{u})R_{u}\sigma_{u}\d\left\langle W,M\right\rangle_{u}\right]\\ 
&+&\rho\mathbb{E}_{t}\left[\int^{T}_{t}e^{-r(u-t)}\left(\partial^{3}_{x}-\partial^{2}_{x}\right) \Gamma BS(u,X_{u},v_{u})\sigma_{u} \d\left\langle W,R\right\rangle_{u}\right]\\ 
&+&\frac{1}{2}\mathbb{E}_{t}\left[\int^{T}_{t}e^{-r(u-t)}\left(\partial^{4}_{x}-2\partial^{3}_{x}+\partial^{2}_{x}\right)\Gamma BS(u,X_{u},v_{u})\d\left\langle M,R\right\rangle_{u}\right].
\end{eqnarray*}
Applying Lemma \ref{lemaclau} 
and using the definition of $a_{u}$, we obtain
\begin{eqnarray*}
&&\left|\frac{1}{8}\mathbb{E}_{t}\left[\int^{T}_{t}e^{-r(u-t)}\Gamma^{2} BS(u,X_{u},v_{u})\d\left\langle M,M\right\rangle_{u}\right]-\Gamma^{2} BS(t,X_{t},v_{t})R_{t}\right|\\ 
&\leq&\frac{C}{8}\mathbb{E}_{t}\left[\int^{T}_{t}e^{-r(u-t)}\left(\frac{1}{a^7_{u}}+\frac{3}{a^6_{u}}+\frac{3}{a^5_{u}}+\frac{1}{a^4_{u}}\right)R_{u}\d\left\langle M,M\right\rangle_{u}\right]\\ 
&+&\frac{C\rho}{2}\mathbb{E}_{t}\left[\int^{T}_{t}e^{-r(u-t)}\left(\frac{1}{a^6_{u}}+\frac{2}{a^5_{u}}+\frac{1}{a^4_{u}}\right)R_{u}\sigma_{u}\d\left\langle W,M\right\rangle_{u}\right]\\ 
&+&C\rho\mathbb{E}_{t}\left[\int^{T}_{t}e^{-r(u-t)}\left(\frac{1}{a^4_{u}}+\frac{1}{a^3_{u}}\right)\sigma_{u} \d\left\langle W,R\right\rangle_{u}\right]\\ 
&+&\frac{C}{2}\mathbb{E}_{t}\left[\int^{T}_{t}e^{-r(u-t)}\left(\frac{1}{a^5_{u}}+\frac{2}{a^4_{u}}+\frac{1}{a^3_{u}}\right)\d\left\langle M,R\right\rangle_{u}\right].
\end{eqnarray*}
Noting that $a_{u}=\sigma_{u}\phi^{1/2}(u,T,\alpha)$, we have
\begin{eqnarray*}
&&\left|\frac{1}{8}\mathbb{E}_{t}\left[\int^{T}_{t}e^{-r(u-t)}\Gamma^{2} BS(u,X_{u},v_{u})d\left\langle M,M\right\rangle_{u}\right]-\Gamma^{2} BS(t,X_{t},v_{t})R_{t}\right|\\ 
&\leq&\frac{C}{4(2-\alpha)^{2}}\mathbb{E}_{t}\left[\int^{T}_{t}e^{-r(u-t)}\left(\frac{\sigma_{u}}{\phi^{3/2}(u,T,\alpha)}+\frac{3\sigma^{2}_{u}}{\phi(u,T,\alpha)}+\frac{3\sigma^{3}_{u}}{\phi^{1/2}(u,T,\alpha)}+\sigma^{4}_{u}\right)\chi_2(u,T,\alpha)  \d u\right]\\ 
&+&\frac{C\rho}{2(2-\alpha)^{2}\xi}\mathbb{E}_{t}\left[\int^{T}_{t}e^{-r(u-t)}\left(\frac{\sigma_{u}}{\phi^{2}(u,T,\alpha)}+\frac{2\sigma^{2}_{u}}{\phi^{3/2}(u,T,\alpha)}+\frac{\sigma^{3}_{u}}{\phi(u,T,\alpha)}\right)\chi_2(u,T,\alpha) \d u\right]\\
&+&\frac{2C\rho }{(2-\alpha)^{2}\xi}\mathbb{E}_{t}\left[\int^{T}_{t}e^{-r(u-t)}\left(\frac{\sigma_{u}}{\phi^{2}(u,T,\alpha)}+\frac{\sigma^{2}_{u}}{\phi^{3/2}(u,T,\alpha)}\right) \chi_2(u,T,\alpha) \d u\right]\\ 
&+&\frac{2C }{(2-\alpha)^{2}}\mathbb{E}_{t}\left[\int^{T}_{t}e^{-r(u-t)}\left(\frac{\sigma_{u}}{\phi^{3/2}(u,T,\alpha)}+\frac{2\sigma^{2}_{u}}{\phi(u,T,\alpha)}+\frac{\sigma^{3}_{u}}{\phi^{1/2}(u,T,\alpha)}\right)\chi_2(u,T,\alpha) \d u\right].
\end{eqnarray*}
Being $\sigma_{u}$ the only stochastic component, we can get the expectation inside. Each power of $\sigma_{u}$ has a different forward value, in this case, we can bound all terms by $\exp\left\{8\xi^{2} (u-t) \right\}$. We have that
\begin{eqnarray*}
&&\left|\frac{1}{8}\mathbb{E}_{t}\left[\int^{T}_{t}e^{-r(u-t)}\Gamma^{2} BS(u,X_{u},v_{u})\d\left\langle M,M\right\rangle_{u}\right]-\Gamma^{2} BS(t,X_{t},v_{t})R_{t}\right|\\ 
&\leq&\frac{C}{4(2-\alpha)^{2} \xi^{2}}\int^{T}_{t}e^{-r(u-t)}\exp\left\{8\xi^{2} (u-t) \right\}\\
&\qquad&\left(\frac{\sigma_{t}}{\phi^{3/2}(u,T,\alpha)}+\frac{3\sigma^{2}_{t}}{\phi(u,T,\alpha)}+\frac{3\sigma^{3}_{t}}{\phi^{1/2}(u,T,\alpha)}+\sigma^{4}_{u}\right)\chi_2(u,T,\alpha) \d u\\ 
&+&\frac{C\rho}{2(2-\alpha)^{2}\xi^3}\int^{T}_{t}e^{-r(u-t)}\exp\left\{8\xi^{2} (u-t) \right\}\\
&\qquad&\left(\frac{\sigma_{t}}{\phi^{2}(u,T,\alpha)}+\frac{2\sigma^{2}_{t}}{\phi^{3/2}(u,T,\alpha)}+\frac{\sigma^{3}_{t}}{\phi(u,T,\alpha)}\right)\chi_2(u,T,\alpha)  \d u\\ 
&+&\frac{2C\rho }{(2-\alpha)^{2}\xi^3}\int^{T}_{t}e^{-r(u-t)}\exp\left\{8\xi^{2} (u-t) \right\}\\
&\qquad&\left(\frac{\sigma_{t}}{\phi^{2}(u,T,\alpha)}+\frac{\sigma^{2}_{t}}{\phi^{3/2}(u,T,\alpha)}\right) \chi_2(u,T,\alpha)\d u\\ 
&+&\frac{2C }{(2-\alpha)^{2}\xi^{2}}\int^{T}_{t}e^{-r(u-t)}\exp\left\{8\xi^{2} (u-t) \right\}\\
&\qquad&\left(\frac{\sigma_{t}}{\phi^{3/2}(u,T,\alpha)}+\frac{2\sigma^{2}_{t}}{\phi(u,T,\alpha)}+\frac{\sigma^{3}_{t}}{\phi^{1/2}(u,T,\alpha)}\right)\chi_2(u,T,\alpha) \d u.
\end{eqnarray*}
Substituting $\phi(u,T,\alpha)$ and using the upper-bound for $\chi_2(u,T,\alpha)$ when $\alpha\geq0$, we have
\begin{eqnarray*}
&&\left|\frac{1}{8}\mathbb{E}_{t}\left[\int^{T}_{t}e^{-r(u-t)}\Gamma^{2} BS(u,X_{u},v_{u})\d\left\langle M,M\right\rangle_{u}\right]-\Gamma^{2} BS(t,X_{t},v_{t})R_{t}\right|\\ 
&\leq&\frac{C}{96(2-\alpha)^{2} \xi^{4}}\int^{T}_{t}e^{-r(u-t)}\exp\left\{8\xi^{2} (u-t) \right\}\\
&\qquad&\Bigl(\sigma_{t}\frac{(2-\alpha)^\frac{3}{2}\xi^{3}}{\left[\exp\left\{(2-\alpha)\xi^{2}  (T-u)\right\}-1\right]^\frac{3}{2}}\\
&\qquad& +3\sigma^{2}_{t}\frac{(2-\alpha)\xi^{2}}{\left[\exp\left\{(2-\alpha)\xi^{2}  (T-u)\right\}-1\right]}+3\sigma^{3}_{t}\frac{(2-\alpha)^\frac{1}{2}\xi}{\left[\exp\left\{(2-\alpha)\xi^{2}  (T-u)\right\}-1\right]^\frac{1}{2}}+\sigma^{4}_{u}\Bigr)\\
&\qquad&\left(\exp\left\{2\xi^{2}  (T-u)\right\}-1\right)^{3}\left(\exp\left\{2\xi^{2}  (T-u)\right\}+3\right) \d u\\ \nonumber
&+&\frac{C\rho}{48(2-\alpha)^{2}\xi^5}\int^{T}_{t}e^{-r(u-t)}\exp\left\{8\xi^{2} (u-t) \right\}\\
&\qquad&\Bigl(\sigma_{t}\frac{(2-\alpha)^{2}\xi^{4}}{\left[\exp\left\{(2-\alpha)\xi^{2}  (T-u)\right\}-1\right]^{2}}+2\sigma^{2}_{t}\left(\frac{(2-\alpha)^\frac{3}{2}\xi^{3}}{\left[\exp\left\{(2-\alpha)\xi^{2}  (T-u)\right\}-1\right]}\right)^\frac{3}{2}\\
&\qquad&+\sigma^{3}_{t}\frac{(2-\alpha)\xi^{2}}{\left[\exp\left\{(2-\alpha)\xi^{2}  (T-u)\right\}-1\right]}\Bigr)\\
&\qquad&\left(\exp\left\{2\xi^{2}  (T-u)\right\}-1\right)^{3}\left(\exp\left\{2\xi^{2}  (T-u)\right\}+3\right)  \d u\\ \nonumber
&+&\frac{C\rho }{12(2-\alpha)^{2}\xi^5}\int^{T}_{t}e^{-r(u-t)}\exp\left\{8\xi^{2} (u-t) \right\}\\
&\qquad&\Bigl(\sigma_{t}\left(\frac{(2-\alpha)^{2}\xi^{4}}{\left[\exp\left\{(2-\alpha)\xi^{2}  (T-u)\right\}-1\right]^{2}}\right)+\sigma^{2}_{t}\left(\frac{(2-\alpha)^\frac{3}{2}\xi^{3}}{\left[\exp\left\{(2-\alpha)\xi^{2}  (T-u)\right\}-1\right]^\frac{3}{2}}\right)\Bigr)\\
&\qquad& \left(\exp\left\{2\xi^{2}  (T-u)\right\}-1\right)^{3}\left(\exp\left\{2\xi^{2}  (T-u)\right\}+3\right) \d u\\ \nonumber
&+&\frac{C }{12(2-\alpha)^{2}\xi^{4}}\int^{T}_{t}e^{-r(u-t)}\exp\left\{8\xi^{2} (u-t) \right\}\\
&\qquad&\Bigl(\sigma_{t}\left(\frac{(2-\alpha)^\frac{3}{2}\xi^{3}}{\left[\exp\left\{(2-\alpha)\xi^{2}  (T-u)\right\}-1\right]^\frac{3}{2}}\right)+2\sigma^{2}_{t}\frac{(2-\alpha)\xi^{2}}{\left[\exp\left\{(2-\alpha)\xi^{2}  (T-u)\right\}-1\right]}\\
&\qquad&+\sigma^{3}_{t}\left(\frac{(2-\alpha)^\frac{1}{2}\xi}{\left[\exp\left\{(2-\alpha)\xi^{2}  (T-u)\right\}-1\right]^\frac{1}{2}}\right)\Bigr)\\
&\qquad&\left(\exp\left\{2\xi^{2}  (T-u)\right\}-1\right)^{3}\left(\exp\left\{2\xi^{2}  (T-u)\right\}+3\right) \d u.
\end{eqnarray*}
The above upper-bound error is difficult to interpret. In order to this, we do a Taylor analysis of one term. Then, the following error behavior is retrieved
$$C \frac{ \xi  \sigma_{t}  T^2}{15 (\alpha -2)^4} \left(5 \left(20 \rho +\sqrt{2}\right)+32 \sqrt{2} \sqrt{T} \sqrt{-(\alpha -2) \xi ^2}\right).$$
\clearpage

\bibliographystyle{references/styles/jp}


\end{document}